\documentclass[10pt]{article}

\usepackage{graphicx, amssymb, latexsym, amsfonts, amsmath, lscape, amscd,
amsthm, color, epsfig, mathrsfs, tikz, enumerate}
\usepackage{fullpage}
\usepackage{graphicx}
\usepackage{epstopdf}
\usepackage{float}
\usepackage{subfloat}
\usepackage{caption}
\usepackage{pstricks}
\usepackage{subfig}

\newtheorem{theorem}{Theorem}[section]

\newtheorem{lemma}[theorem]{Lemma}

\newtheorem{observation}[theorem]{Observation}

\begin{document}

\title{{\bf Pushable chromatic number of graphs \\ with degree constraints\thanks{The authors were partly supported by ANR project HOSIGRA (ANR-17-CE40-0022) and by IFCAM project ``Applications of graph homomorphisms'' (MA/IFCAM/18/39).}}}
\author{
{\sc Julien Bensmail}$^{a}$, {\sc Sandip Das}$\,^b$, {\sc Soumen Nandi}$\,^{c}$, {\sc Th\'eo Pierron}$\,^{d,f}$,\\
{\sc Soumyajit Paul}$^{c}$, {\sc Sagnik Sen}$^{e}$, {\sc \'Eric Sopena}$^{f}$\\
\mbox{}\\
{\small $(a)$ Universit\'e C\^ote d'Azur, Inria, CNRS, I3S, France}\\
{\small $(b)$ Indian Statistical Institute, Kolkata, India}\\
{\small $(c)$ Institute of Engineering \& Management, Kolkata, India}\\
{\small $(d)$ Masaryk University, Brno, Czech Republic}\\
{\small $(e)$ Indian Institute of Technology - Dharwad, Dharwad, India}\\
{\small $(f)$ Univ. Bordeaux, CNRS, Bordeaux INP, LaBRI, UMR 5800, F-33400 Talence, France}\\}

\date{\today}

\maketitle

\begin{abstract}
Pushable homomorphisms and the pushable chromatic number $\chi_p$ of oriented graphs were introduced by Klostermeyer and MacGillivray in 2004. 
They notably observed that,
for any oriented graph $\overrightarrow{G}$, we have 
$\chi_p(\overrightarrow{G}) \leq \chi_o(\overrightarrow{G}) \leq 2 \chi_p(\overrightarrow{G})$, where $\chi_o(\overrightarrow{G})$ denotes the oriented chromatic number of $\overrightarrow{G}$.
This stands as first general bounds on $\chi_p$. This parameter was further studied in later works.

This work is dedicated to the pushable chromatic number of oriented graphs fulfilling particular degree conditions. For all $\Delta \geq 29$, we first prove that the maximum value of the pushable chromatic number of an oriented graph with maximum degree $\Delta$ lies between $2^{\frac{\Delta}{2}-1}$ and $(\Delta-3) \cdot (\Delta-1) \cdot  2^{\Delta-1} + 2$ which implies an improved bound on the oriented chromatic number of the same family of graphs. For subcubic oriented graphs, that is, when $\Delta \leq 3$, we then prove that the maximum value of the pushable chromatic number is~$6$ or~$7$.
We also prove that the maximum value of the pushable chromatic number of oriented graphs with maximum average degree less than~$3$ lies between~$5$ and~$6$. The former upper bound of~$7$ also holds as an upper bound on the pushable chromatic number of planar oriented graphs with girth at least~$6$.

\medskip

\noindent \textit{Keywords:} oriented coloring, push operation, 
graph homomorphism, maximum degree, subcubic graph, 
maximum average degree  

\end{abstract}

\section{Introduction and main results}


An \textit{oriented graph} is a loopless directed graph without opposite arcs. Equivalently, an oriented graph $\overrightarrow{G}$ can be seen as an \textit{orientation} of a simple undirected graph $G$. Throughout this paper, we stick to the notation from the previous sentence. That is, we always refer to an oriented graph $\overrightarrow{G}$ using an arrow symbol, which makes apparent that $\overrightarrow{G}$ is an orientation of $G$. We denote by $V(G)$ and $E(G)$ the sets of vertices and edges of $G$, respectively, while we denote by $V(\overrightarrow{G})$ and $A(\overrightarrow{G})$ the sets of vertices and arcs of $\overrightarrow{G}$, respectively. Also, when referring to a notation, notion or term for $\overrightarrow{G}$ that is usually defined for undirected graphs, we implicitly refer to the corresponding notation, notion or term regarding $G$.

The notions of \textit{oriented coloring} and \textit{oriented chromatic number} of oriented graphs were introduced by 
Courcelle~\cite{courcelle-monadic} in 1994,
and have been intensively studied since then 
(see the recent survey~\cite{sopena_survey} for more details). 
One way of defining these notions is through the notion of \textit{graph homomorphisms}. 
For two oriented graphs $\overrightarrow{G}$ and $\overrightarrow{H}$,
a \textit{homomorphism} from $ \overrightarrow{G}$ to $ \overrightarrow{H}$ is a
 mapping $\phi: V(\overrightarrow{G}) \rightarrow V(\overrightarrow{H})$ 
 such that $uv \in A(\overrightarrow{G})$ implies 
 $\phi(u)\phi(v) \in A(\overrightarrow{H})$. 
 We write $ \overrightarrow{G} \rightarrow  \overrightarrow{H}$ whenever 
 a homomorphism from $\overrightarrow{G}$ to $ \overrightarrow{H}$ exists.   
 The \textit{oriented chromatic number} $\chi_o( \overrightarrow{G})$ of  
 $\overrightarrow{G}$ is 
  the minimum order (number of vertices) of an oriented graph $\overrightarrow{H}$ such that
$\overrightarrow{G} \rightarrow \overrightarrow{H}$.
  
\medskip

In 2004, Klostermeyer and
 MacGillivray~\cite{push} introduced the \textit{pushable chromatic number} of oriented graphs. \textit{Pushing} a vertex $v$ of an oriented graph $\overrightarrow{G}$ means changing the  orientation of all arcs incident with $v$, i.e., replacing every arc $vu$ by the arc $uv$, and vice versa.
   Two  oriented graphs $\overrightarrow{G}$
 and $\overrightarrow{G}'$ are \textit{in a push relationship} if $\overrightarrow{G}'$ can be obtained from $\overrightarrow{G}$ by pushing some vertices 
 of $\overrightarrow{G}$.
 Note that  being in push relationship is  an
equivalence relation. The class of the oriented graphs that are in a push relationship with $\overrightarrow{G}$ is denoted by  $[\overrightarrow{G}]$. 
Observe that any two oriented graphs from $[\overrightarrow{G}]$ have the same underlying graph, which is $G$.
 
For two oriented graphs $\overrightarrow{G}$ and $\overrightarrow{H}$, a \textit{pushable homomorphism} from $ \overrightarrow{G}$ to $ \overrightarrow{H}$ is a
 mapping $\phi: V(\overrightarrow{G}) \rightarrow V(\overrightarrow{H})$ such that 
 there exists  
$\overrightarrow{G}' \in [\overrightarrow{G}]$ 
for which $\phi$ is a homomorphism from $\overrightarrow{G}'$ to $\overrightarrow{H}$.  
 We write $\overrightarrow{G} \xrightarrow{{\rm push}}  \overrightarrow{H}$ whenever there exists a 
 pushable homomorphism from $ \overrightarrow{G}$ to $ \overrightarrow{H}$.
 The \textit{pushable chromatic number} $\chi_p( \overrightarrow{G})$ of  $\overrightarrow{G}$ is 
  the minimum order of an oriented graph $\overrightarrow{H}$ such that
$\overrightarrow{G} \xrightarrow{{\rm push}} \overrightarrow{H}$.
 
The seminal work of Klostermeyer and MacGillivray on these notions opened the way to more works on the topic. For instance, results on the pushable chromatic number can be found in ~\cite{nandi-bensmail-push1,ochem-push,sen-push}, while the  push operation was further studied in~\cite{fisher-push, klostermeyer1, klostermeyer2, Garyandwood, Mosesian, Pretzel3, Pretzel2, Pretzel1}. Some complexity issues related to pushable homomorphisms were studied 
in~\cite{ochem-push,push}. Regarding our investigations in this paper, an important result from the seminal work~\cite{push} of Klostermeyer and MacGillivray is the following general relation between $\chi_o$ and $\chi_p$.
 
\begin{theorem}[Klostermeyer, MacGillivray~\cite{push}]\label{theorem:relation-chio-chip}
For every oriented graph $\overrightarrow{G}$, we have 
$$\chi_p(\overrightarrow{G}) \leq \chi_o(\overrightarrow{G}) \leq 2 \chi_p(\overrightarrow{G}).$$
\end{theorem}
 
Theorem~\ref{theorem:relation-chio-chip} yields another point for studying the pushable chromatic number of oriented graphs, as it is a way to get bounds on the oriented chromatic number. Sen, in~\cite{sen-push}, also established a strong connection between pushable homomorphisms and oriented homomorphisms of oriented graphs.

\medskip

The notions of oriented chromatic number and pushable chromatic number can also be extended to undirected graphs $G$ by setting 
 $$\chi_{o}(G)= \max\{\chi_o(\overrightarrow{G}): \overrightarrow{G} \text{ is an orientation of } G\}$$ and
 $$\chi_{p}(G)= \max\{\chi_p(\overrightarrow{G}): \overrightarrow{G} \text{ is an orientation of } G\}.$$
 
\medskip

A natural question is, given a family $\mathcal{F}$ of undirected graphs, how large can the oriented chromatic number and the pushable chromatic number of its members be? In other words, we are interested in the two parameters $\chi_{o}(\mathcal{F})= \max\{\chi_o(G) : G \in \mathcal{F}\}$ and $\chi_{p}(\mathcal{F})= \max\{\chi_p(G) : G \in \mathcal{F}\}$. Regarding the pushable chromatic number, partial results were obtained for the families of outerplanar graphs, $2$-trees, planar graphs, planar graphs with girth restrictions, and graphs with bounded acyclic chromatic number (see~\cite{ochem-push, push, sen-push}).
However, to the best of our knowledge, nothing general is known regarding the family $\mathcal{G}_{\Delta}$ of graphs with maximum degree $\Delta $ and the family 
 $\mathcal{G}^{\rm c}_{\Delta}$ of connected graphs with maximum degree $\Delta$. 
 Unlike the ordinary chromatic number, 
 the oriented and pushable chromatic number for the families $\mathcal{G}_{\Delta}$ and 
 $\mathcal{G}^{\rm c}_{\Delta}$  can be different.  
 Finding the oriented and pushable chromatic number $\mathcal{G}^{\rm c}_{\Delta}$ is our main concern in this paper.

\bigskip
We thus initiate the study of the pushable chromatic number of  $\mathcal{G}^{\rm c}_{\Delta}$. Adapting a probabilistic proof used by 
Kostochka, Sopena and Zhu in~\cite{Kostochka97acyclicand}, we first provide general bounds for large enough $\Delta$.

\begin{theorem}\label{Push_chrom_th main}
For all $\Delta \geq 29$, we have $$2^{\frac{\Delta}{2} - 1} \leq \chi_p(\mathcal{G}^{\rm c}_{\Delta}) \leq (\Delta-3) \cdot (\Delta-1) \cdot 2^{\Delta-1} +2.$$
\end{theorem}

Note that the lower bound and the upper bound in Theorem~\ref{Push_chrom_th main} are both exponential in $\Delta$. Also, it is worth mentioning that the upper bound established in Theorem~\ref{Push_chrom_th main} is better than the upper bound that one would directly get from Theorem~\ref{theorem:relation-chio-chip} and the best upper bound on $\chi_o(\mathcal{G}_{\Delta})$ to date, 
which is that $\chi_o(\mathcal{G}_{\Delta}) \leq 2\Delta^2 \cdot 2^{\Delta}$ (see~\cite{Kostochka97acyclicand}). Actually, employing another trick used by Duffy in~\cite{duffy}, Theorem~\ref{Push_chrom_th main} also yields the following improved upper bound on $\chi_o(\mathcal{G}_{\Delta}^c)$ as a side result.

\begin{theorem}\label{Push_chrom_th main2}
For all $\Delta \geq 29$, we have
$$2^{\frac{\Delta}{2}} \leq \chi_o(\mathcal{G}^{\rm c}_{\Delta}) \leq (\Delta-3) \cdot (\Delta-1) \cdot 2^{\Delta} +2.$$
\end{theorem}

When it comes to coloring graphs with given maximum degree, a natural step to make is considering graphs with low maximum degree. This concern is actually a major one regarding oriented coloring, as it is still open what the value of $\chi_o(\mathcal{G}^{\rm c}_{3})$ is. Sopena~\cite{sopena_survey} conjectured that 
$\chi_o(\mathcal{G}^{\rm c}_{3})=7$, and, to date, we know that 
$7 \leq \chi_o(\mathcal{G}_{3}) \leq 9$ and $7 \leq \chi_o(\mathcal{G}^{\rm c}_{3}) \leq 8$
hold
(see \cite{duffy-cubic,sopena_survey} and~\cite{duffy-connected-8}, respectively).

Due to the general connection between the oriented chromatic number and the pushable chromatic number, it makes sense wondering about $\chi_p(\mathcal{G}^{\rm c}_{3})$ as well. In this work, we provide the following result as a first step towards this question.

\begin{theorem}\label{th subcubic}
We have $6 \leq \chi_p(\mathcal{G}^{\rm c}_{3}) \leq \chi_p(\mathcal{G}_{3}) \leq 7$.
\end{theorem}

In graph coloring theory, another relevant aspect related to the vertex degrees is the maximum average degree. 
Precisely, the \textit{maximum average degree} ${\rm mad}(G)$ of a graph $G$ is 
$${\rm mad}(G) = \max\left\{\frac{2|E(H)|}{|V(H)|} \text{ : } H \text{ is a subgraph of } G\right\}.$$
In this work, we also study the pushable chromatic number of the family 
$\mathcal{G}^{{\rm mad}}_3 = \{G : {\rm mad}(G) < 3\}$ of graphs with maximum average degree less than~$3$.
Our main result reads as follows.

\begin{theorem}\label{th mad3}
We have $ 5 \leq  \chi_p(\mathcal{G}^{{\rm mad}}_{3}) \leq 7$.
\end{theorem}

It was previously proved in~\cite{sen-push} that for the family $\mathcal{G}^{{\rm mad}}_{8/3} = \{G : {\rm mad}(G) < \frac{8}{3}\}$ 
we have 
$\chi_p(\mathcal{G}^{{\rm mad}}_{8/3}) =4$.  
More precisely, in that result the equality follows from the existence of planar graphs with girth $8$ and pushable chromatic number~$4$.
This, and, because planar graphs with girth at least~$6$ have maximum average degree strictly less than $3$,
Theorem~\ref{th mad3} yield the following, where $\mathcal{P}_6$ denotes the family of planar graphs with girth at least~$6$. 

\begin{theorem}\label{th planar g6}
We have  $4 \leq \chi_{p} ( \mathcal{P}_6) \leq 7$.
\end{theorem}

This paper is organized as follows. We start off by introducing, in Section~\ref{sec pre}, some notation, terminology, and preliminary results.
The next sections are devoted to proving Theorems~\ref{Push_chrom_th main} and~\ref{Push_chrom_th main2} (Section~\ref{sec 1}),
Theorem~\ref{th subcubic} (Section~\ref{sec 2}), and Theorem~\ref{th mad3} (Section~\ref{sec 3}).
Open questions and perspectives for future work 
are discussed in Section~\ref{section:conclusion}.

\section{Notation, terminology, and preliminary results}\label{sec pre}

For an arc $uv$ of an oriented graph $\overrightarrow{G}$, we say that $u$ is
a \textit{$-$-neighbor} of $v$ while $v$ is a \textit{$+$-neighbor} of $u$. 
The set of the $-$-neighbors ($+$-neighbors, respectively) of any vertex $v$ of $\overrightarrow{G}$ is denoted by $N^-(v)$ ($N^+(v)$, respectively).
For a set $S$ of vertices of $\overrightarrow{G}$ and some $\alpha \in \{-,+\}$, we define $N^{\alpha}(S)= \bigcup_{v \in S}N^{\alpha}(v)$.

\medskip

To prove that all oriented graphs from a given family admit homomorphisms to a given oriented graph $\overrightarrow{H}$,
we generally need $\overrightarrow{H}$ to have very strong properties.
In most of the proofs from the literature on the topic, and in our proofs in the current paper as well,
a strong property we consider is the possibility, given a partial homomorphism from an oriented graph $\overrightarrow{G}$ to $\overrightarrow{H}$,
to extend the partial homomorphism to another vertex $v$ of $\overrightarrow{G}$, assuming some of its neighbors (which can be in any of $N^-(v)$ and $N^+(v)$) have already been assigned an image.
A way to define this intuition is through the notion of Property $P(j,k)$, which we define formally in what follows.

A \textit{$j$-vector} $\vec{a} = (a_1, \dots, a_j)$ is a vector where $a_i \in \{-,+\}$ for every $i \in \{1,\dots,j\}$.
We denote by $\vec{a}^c = (a_1^c, \dots, a_j^c)$ the $j$-vector where $a_i^c \neq a_i$ for every $i \in \{1,\dots,j\}$.
Let $J = \{v_1, \dots, v_j\}$ be a set of $j$ vertices of $V(\overrightarrow{G})$. 
Then we define the set 
$$N^{\vec{a}}(J) = \left\{v \in V(\overrightarrow{G}) : v \in N^{a_i}(v_i) \text{ for all } 1 \leq i \leq j\right\}
\cup  \left\{v \in V(\overrightarrow{G}) : v \in N^{a_i^c}(v_i) \text{ for all } 1 \leq i \leq j\right\}.$$
Observe that $N^{\vec{a}}(J) = N^{\vec{a}^c}(J)$. 
We say that \textit{$\overrightarrow{G}$ has Property $P(j,k)$} if for every $j$-vector $\vec{a}$ and every $j$-set $J$ we have $|N^{\vec{a}}(J)| \geq k$. 

\medskip

A bijective homomorphism whose inverse is also a homomorphism is an \textit{isomorphism}.   
An oriented graph $\overrightarrow{G}$  is  \textit{vertex-transitive} 
if for every two vertices $u,v \in V(\overrightarrow{G})$ there is an isomorphism $f$ of $\overrightarrow{G}$
 such that $f(u) = v$. 
We also say that $\overrightarrow{G}$  is  \textit{arc-transitive} 
if given any two arcs $uv,xy \in A(\overrightarrow{G})$ it is possible to find an 
isomorphism $f$  of $\overrightarrow{G}$ such that $f(u) = x$ and $f(v)=y$. 

\begin{figure}
\centering
	\begin{tikzpicture}[inner sep=0.7mm]
	
	\foreach \a in {0,...,6}{
		\node[draw, circle, line width=1pt](v\a) at (\a*360/7:3cm){$\a$};	
	}
	
	\draw[-latex,line width=1pt,black] (v0) -- (v1);
	\draw[-latex,line width=1pt,black] (v1) -- (v2);
	\draw[-latex,line width=1pt,black] (v2) -- (v3);
	\draw[-latex,line width=1pt,black] (v3) -- (v4);
	\draw[-latex,line width=1pt,black] (v4) -- (v5);
	\draw[-latex,line width=1pt,black] (v5) -- (v6);
	\draw[-latex,line width=1pt,black] (v6) -- (v0);
	
	\draw[-latex,line width=1pt,black] (v0) -- (v2);
	\draw[-latex,line width=1pt,black] (v1) -- (v3);
	\draw[-latex,line width=1pt,black] (v2) -- (v4);
	\draw[-latex,line width=1pt,black] (v3) -- (v5);
	\draw[-latex,line width=1pt,black] (v4) -- (v6);
	\draw[-latex,line width=1pt,black] (v5) -- (v0);
	\draw[-latex,line width=1pt,black] (v6) -- (v1);
	
	\draw[-latex,line width=1pt,black] (v0) -- (v4);
	\draw[-latex,line width=1pt,black] (v1) -- (v5);
	\draw[-latex,line width=1pt,black] (v2) -- (v6);
	\draw[-latex,line width=1pt,black] (v3) -- (v0);
	\draw[-latex,line width=1pt,black] (v4) -- (v1);
	\draw[-latex,line width=1pt,black] (v5) -- (v2);
	\draw[-latex,line width=1pt,black] (v6) -- (v3);

 	\end{tikzpicture}
 	\caption{The Paley tournament ${\rm Pal}_7$ on seven vertices.}
 	\label{figure:pal7}
 \end{figure}
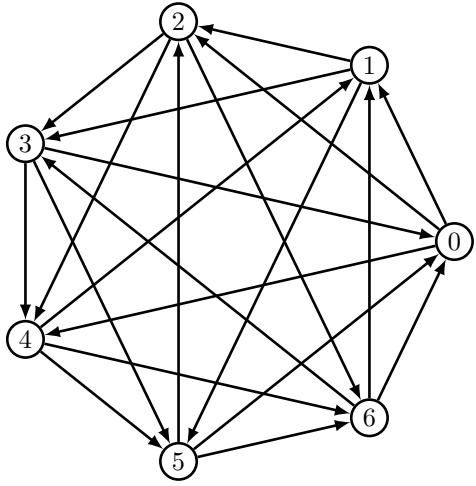

In the context of oriented homomorphisms and pushable homomorphisms, \textit{Paley tournaments} stand, due to their very regular structure,
as good candidates 
to map families of oriented graphs to.
In this work, our upper bounds in Theorems~\ref{th subcubic} and~\ref{th mad3} are actually obtained via pushable homomorphisms to $\overrightarrow{{\rm Pal}_7}$, the Paley tournament on seven vertices.
$\overrightarrow{{\rm Pal}_7}$ (depicted in Figure~\ref{figure:pal7}) is the oriented graph (tournament) with vertex set $\mathbb{Z}/7\mathbb{Z} = \{0,1,\dots, 6\}$ in which $ij$ is an arc if and only if $j-i$ is a nonzero square in $\mathbb{Z}/7\mathbb{Z}$ (where, here and further, all operations involving vertices of $\overrightarrow{{\rm Pal}_7}$ are understood modulo~$7$). 
In other words, $ij$ is an arc if and only if $j-i \in \{1,2,4\}$.
  
  In this work, we will make use of the following properties of interest of $\overrightarrow{{\rm Pal}_7}$.
   
   \begin{lemma}[Marshall~\cite{marshall17}]\label{P7 is vt-at}
   $\overrightarrow{{\rm Pal}_7}$ is vertex-transitive and arc-transitive. 
   \end{lemma}

 \begin{lemma}[Marshall~\cite{marshall17}]\label{P7 P22}
  $\overrightarrow{{\rm Pal}_7}$ has Properties $P(1,6)$ and $P(2,2)$. 
   \end{lemma}
   
   We also note the following other interesting property of $\overrightarrow{{\rm Pal}_7}$.  
   
   \begin{observation}
   For every $i \in V(\overrightarrow{{\rm Pal}_7})$ and $\alpha, \beta, \gamma \in \{+,-\}$, we have
   \begin{equation}\label{eq first neighborhood}
N^\alpha(N^{\beta}(i)) =
\begin{cases}
V(\overrightarrow{{\rm Pal}_7}) \setminus \{i\} & \text{ if } \alpha = \beta \\
V(\overrightarrow{{\rm Pal}_7}) & \text{ if } \alpha \neq \beta
\end{cases}
\end{equation}
and
\begin{equation}\label{eq second neighborhood}
N^\alpha(N^{\beta}(N^{\gamma}(i))) = V(\overrightarrow{{\rm Pal}_7}).
\end{equation}
   \end{observation}

\begin{proof}
 As $\overrightarrow{{\rm Pal}_7}$ is vertex-transitive, it is enough to verify the above equations for $i=0$,
 which can easily be done by hand.
\end{proof}

\section{Proofs of Theorems~\ref{Push_chrom_th main} and~\ref{Push_chrom_th main2}}\label{sec 1}

Let $t$ be a fixed integer. For a given integer $j$, we set $f_t(j) = (t-j)(t-2)+1$. 
In the next result, we show that if an oriented graph $\overrightarrow{G}$ has 
Property $P(t-1, f_t(t-1))$ for some $t$, then it also has Property $P(j, f_t(j))$ for all $j \in \{0,1,\dots,t-1\}$.

\begin{lemma}\label{Push_chrom_key-lemma-recursive}
If an oriented graph $\overrightarrow{G}$ has Property
$P(t-1, f_t(t-1))$ for some $t$, then it also has Property $P(j, f_t(j))$ for all $j \in \{0,1,\dots,t-1\}$.
\end{lemma}

\begin{proof}
Suppose $\overrightarrow{G}$ has Property $P(j, k)$. 
Now consider any $(j-1)$-vector 
$\vec{a}' = (a_1, \dots, a_{j-1})$, any $(j-1)$-set $J' = \{v_1, \dots, v_{j-1}\}$,
and a vertex $v_j$ of $\overrightarrow{G}$ not in $J'$. 
Let $\vec{a}_+ = (a_1, \dots, a_{j-1}, +)$, $\vec{a}_- = (a_1, \dots, a_{j-1}, -)$ and 
$J = \{v_1, \dots, v_{j-1}, v_j\}$. Since $\overrightarrow{G}$ has Property $P(j, k)$, we must have 
$|N^{\vec{a}_+}(J)| \geq k$ and $|N^{\vec{a}_-}(J)| \geq k$. Observe that 
$N^{\vec{a}_+}(J) \cap N^{\vec{a}_-}(J) = \emptyset$ and 
that $N^{\vec{a}_+}(J),   N^{\vec{a}_-}(J)  \subseteq N^{\vec{a}'}(J')$. 
Thus $\overrightarrow{G}$ has Property $P(j-1, 2k)$.   
We are now done by induction since $f_t(j-1) \leq 2f_t(j)$ for all $j \in \{0,1,\dots,t-1\}$.
\end{proof}

We now prove the existence of tournaments having Property $P(j,k)$ for particular values of $j$ and $k$.

\begin{lemma}\label{Push_chrom_key-lemma}
For all $t \geq 29$, there exist tournaments with Property $P(t-1, t-1)$ and order 
$$c = (t-3) \cdot (t-1) \cdot 2^{t-1}.$$
\end{lemma} 

\begin{proof}
Let $\overrightarrow{C}$ be a random tournament in which every arc is oriented in one way or the other with equal probability $\frac{1}{2}$.
We show below that the probability that $\overrightarrow{C}$ does not have Property $P(t-1, t-1)$ is strictly less than~$1$ 
when  
$|\overrightarrow{C}| = c = (t-3) \cdot (t-1) \cdot 2^{t-1}$. 
Let $\mathbb{P}(J,\vec{a})$ denote the probability that the \textit{bad event}
$|N^{\vec{a}}(J)| < f_t(t-1) = (t-2)+1$ occurs, where $J$ is a $(t-1)$-set of $\overrightarrow{C}$ and $\vec{a}$ is a $(t-1)$-vector. Then

\begin{equation*} \label{Push_chrom_Signed_eq1}
\begin{split}
\mathbb{P}(J, \vec{a})&\leq \sum_{\underset{S\cap J=\emptyset}{|S|\leq t-2}} \mathbb{P}(S=N^{\vec{a}}(J)) = \sum_{\underset{S\cap J=\emptyset}{|S|\leq t-2}} \prod_{x\in S} \mathbb{P}(x\in N^{\vec{a}}(J)) \cdot\prod_{x\notin J\cup S} \mathbb{P}(x\notin N^{\vec{a}}(J))\\
&=\sum_{\underset{S\cap J=\emptyset}{|S|\leq t-2}} 2\cdot 2^{-|S||J|}\cdot (1-2\cdot 2^{-|J|})^{c-|J|-|S|}\\
&=\sum\limits_{i=0}^{t-2} {c-(t-1) \choose i} \cdot 2^{-i(t-2)} \cdot \left(1 - 2^{-(t-2)}\right)^{c - i - (t-1)} \\
&= \left(1 - 2^{-(t-2)}\right)^c \cdot \sum\limits_{i=0}^{t-2} \frac{c^i}{i!} \cdot 2^{-i(t-2)}\cdot\left(\frac{2^{t-2}}{2^{t-2}-1}\right)^{i+t-1}  \\
&< e^{-c2^{-(t-2)}} \cdot \sum\limits_{i=0}^{t-2} c^i \frac{2^{(t-2)(t-1)}}{(2^{t-2}-1)^{i+t-1}}\leq e^{-c2^{-(t-2)}} \cdot \frac{2^{(t-2)(t-1)}}{(2^{t-2}-1)^{t-1}}\cdot \sum\limits_{i=0}^{t-2} c^i \\
 &< 2 e^{-c2^{-(t-2)}} \cdot \frac{c^{t-1}-1}{c-1}  \leq e^{-c2^{-(t-2)}} \cdot c^{t-1}.
\end{split}
\end{equation*}

Let $\mathbb{P}(\mathcal{B})$ denote the probability that at least one bad event occurs. 
To prove the statement it is then enough to  show that $\mathbb{P}(\mathcal{B}) < 1$. 
Let $T$ denote the set of all $(t-1)$-sets of vertices of $\overrightarrow{C}$, and $W$ denote the set of all $(t-1)$-vectors having $+$ in the first coordinate. Note that given any $(t-1)$-vector $\vec{a}$, 
exactly one of $\vec{a}$ and $\vec{a}^c$ must belong to $W$.  Then 

\begin{equation*} \label{Push_chrom_Signed_eq2}
\begin{split}
\mathbb{P}(\mathcal{B})  &\leq  \sum_{J \in T} \sum_{\vec{a} \in W} \mathbb{P}(J, \vec{a})  
< {c \choose t-1} \cdot 2^{t-2} \cdot e^{-c2^{-(t-2)}} \cdot c^{t-1} \\
&< \frac{c^{t-1}}{(t-1)!} \cdot 2^{t-2} \cdot e^{-c2^{-(t-2)}} \cdot c^{t-1} 
< \frac{2^{t-2}}{(t-1)!} \cdot  e^{-c2^{-(t-2)}} \cdot c^{2(t-1)} \\
 &<  e^{-c2^{-(t-2)}} \cdot c^{2(t-1)}
  <  \left( \frac{(t-3)^2 \cdot (t-1)^2\cdot 2^{2(t-1)}}{e^{2(t-3)}}\right)^{t-1} 
  < 1.
\end{split}
\end{equation*}

In particular, the last inequality follows because $t \geq 29$.
This completes the proof.
\end{proof}

We now show that if a tournament $\overrightarrow{C}$ has Property
$P(j,f_t(j))$ for all $j \in \{1,\dots,\Delta-1\}$ 
where $t = \Delta$, then any connected oriented graph with 
maximum degree $\Delta$ and degeneracy $\Delta-1$ admits a pushable homomorphism to $\overrightarrow{C}$.

\begin{lemma}\label{Push_chrom_key-lemma2}
Let $\overrightarrow{C}$ be an oriented graph 
 having Property 
$P(j,f_t(j))$ for all $j \in \{1,\dots,\Delta-1\}$ where 
$t = \Delta$, and $\overrightarrow{G}$ be a connected oriented graph with 
maximum degree $\Delta$ and degeneracy $\Delta-1$. 
Then $\overrightarrow{G} \xrightarrow{{\rm push}} \overrightarrow{C}$. 
\end{lemma}

\begin{proof}
Let us assume the vertices of $\overrightarrow{G}$ are labeled $v_1, \dots, v_k$ so that each vertex has at most $\Delta-1$ neighbors with smaller index.
For every $l \in \{1,\dots,k\}$, we denote by $\overrightarrow{G}_l$ the  oriented graph induced by the vertices in $\{v_1,  \dots, v_l\}$.
We now inductively construct a  homomorphism 
$g: \overrightarrow{G} \rightarrow \overrightarrow{C}$ with the following properties:

\begin{itemize}
\item For every $l \in \{1,\dots,k\}$, the partial mapping $g(v_1), \dots, g(v_l)$ is a homomorphism from $\overrightarrow{G}_l$ to $\overrightarrow{C}$.

\item For every $i > l$, all neighbors of $v_i$ with index 
at most $l$ have different images by the mapping $g$.  
\end{itemize} 

For $l=1$, consider any  partial mapping $g(v_1)$. 
Suppose now that the function $g$ satisfies the above two properties for all $i \leq l$ for some  fixed $l \in \{1,\dots,k-1\}$. 
Let $A$ be the set of neighbors of $v_{l+1}$   with  index greater than $l+1$,
and $B$ be the set of vertices   with index at most $l$ and with at least one neighbor 
in $A$. 
 Note that  $|B| \leq (\Delta-2)|A|$.

 Let $D$ be the set of possible  options 
 for $g(v_{l+1})$ leading to the partial mapping 
 being a homomorphism from $\overrightarrow{G}_{l+1}$ to $\overrightarrow{C}$. Let $A'$ be the set of neighbors of $v_{l+1}$   with  index less than $l+1$. 
 Therefore, due to Property 
$P(|A'|,f_t(|A'|))$ of $\overrightarrow{C}$, we have 
\[|D| \geq f_t(|A'|)=(\Delta - |A'|)(\Delta - 2) +1 > (\Delta-2)|A| \geq |B|,\] which implies
  $|D|  > |B|$. 
 Thus choose any vertex from $D \setminus B$ as the image $g(v_{l+1})$. 
 Note that the resulting partial mapping satisfies the  two required conditions as well.
 This concludes the proof.
\end{proof}

We are  now ready to prove Theorem~\ref{Push_chrom_th main}. 

\begin{proof}[Proof of Theorem~\ref{Push_chrom_th main}]
The lower bound follows from Theorem~\ref{theorem:relation-chio-chip} and the bound  $2^{\Delta/2} \leq \chi_o(\mathcal{G}_{\Delta})$ established by Kostochka, Sopena and Zhu in~\cite{Kostochka97acyclicand}.
We now focus on proving the upper bound.
Let $\overrightarrow{G}$ be a connected oriented graph with maximum degree $\Delta \geq 29$.
Note that if $\overrightarrow{G}$ is not $\Delta$-regular then $\overrightarrow{G}$ is $(\Delta-1)$-degenerate.
In that case we are done by Lemmas~\ref{Push_chrom_key-lemma} and~\ref{Push_chrom_key-lemma2}. 
So assume $\overrightarrow{G}$ is $\Delta$-regular.
Delete one arc $uv$ from $\overrightarrow{G}$ to obtain a connected oriented graph with maximum degree $\Delta$ and degeneracy $\Delta-1$. 
This new oriented graph admits a pushable homomorphism $g$ to an oriented graph $\overrightarrow{C}$   with Property 
$P(j,f_t(j))$ for all $j \in \{1,\dots,\Delta-1\}$. 
Now add two new vertices $x$ and $y$ to $\overrightarrow{C}$ to obtain a new oriented graph 
$\overrightarrow{\hat{C}}$. Modify the pushable homomorphism $g$ to $\hat{g}$ by setting $\hat{g}(u) = x$, 
$\hat{g}(v) = y$ and $\hat{g}(w) = g(w)$ for all $w \neq u,v$. Moreover,  choose the 
direction of the arcs incident with  $x$ and $y$ 
 in such a way that $\hat{g}$ is a pushable homomorphism from $\overrightarrow{G}$ to $\overrightarrow{\hat{C}}$.
\end{proof}

The proof of Theorem~\ref{Push_chrom_th main} above can also be employed to prove Theorem~\ref{Push_chrom_th main2}.

\begin{proof}[Proof of Theorem~\ref{Push_chrom_th main2}]
The lower bound is due to a result of Kostochka, Sopena and Zhu in~\cite{Kostochka97acyclicand}.
Let us now focus on the upper bound.
From Lemmas~\ref{Push_chrom_key-lemma} and~\ref{Push_chrom_key-lemma2}, we know that 
if $\overrightarrow{G}$ has
maximum degree $\Delta$ and is $(\Delta-1)$-degenerate, then 
$\chi_o(\overrightarrow{G}) \leq 2 \chi_p(\overrightarrow{G}) \leq 2(\Delta-3) \cdot (\Delta-1) \cdot 2^{\Delta}$. 
Thus we are done for all oriented graphs of $\mathcal{G}_{\Delta}$  but the ones that are $\Delta$-regular.
For these oriented graphs, the upper bound can be proved similarly as in the proof of Theorem~\ref{Push_chrom_th main}. 
\end{proof}

\section{Proof of Theorem~\ref{th subcubic}}\label{sec 2}


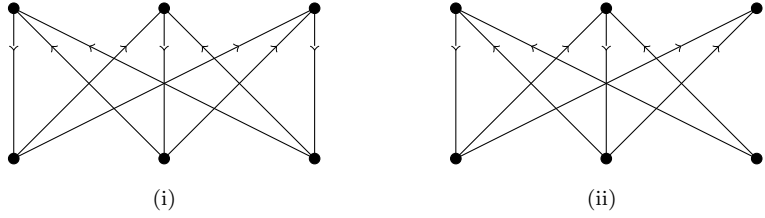
\begin{figure}
\centering
\subfloat[]{
\scalebox{1}{
\begin{tikzpicture}
\filldraw [black] (0,0) circle (2pt) {node[below]{}};
\filldraw [black] (2,0) circle (2pt) {node[above]{}};
\filldraw [black] (4,0) circle (2pt) {node[above]{}};
\filldraw [black] (0,2) circle (2pt) {node[below]{}};
\filldraw [black] (2,2) circle (2pt) {node[above]{}};
\filldraw [black] (4,2) circle (2pt) {node[above]{}};
\draw[-<] (0,0) -- (0,1.5);
\draw[-] (0,1.5) -- (0,2);
\draw[->] (0,0) -- (1.5,1.5);
\draw[-] (1.5,1.5) -- (2,2);
\draw[->] (0,0) -- (3,1.5);
\draw[-] (3,1.5) -- (4,2);
\draw[->] (2,0) -- (.5,1.5);
\draw[-] (.5,1.5) -- (0,2);
\draw[-<] (2,0) -- (2,1.5);
\draw[-] (2,1.5) -- (2,2);
\draw[->] (2,0) -- (3.5,1.5);
\draw[-] (3.5,1.5) -- (4,2);
\draw[->] (4,0) -- (1,1.5);
\draw[-] (1,1.5) -- (0,2);
\draw[->] (4,0) -- (2.5,1.5);
\draw[-] (2.5,1.5) -- (2,2);
\draw[-<] (4,0) -- (4,1.5);
\draw[-] (4,1.5) -- (4,2);
\end{tikzpicture}
}
}
\hspace{30pt}
\subfloat[]{
\scalebox{1}{
\begin{tikzpicture}
\filldraw [black] (6,0) circle (2pt) {node[below]{}};
\filldraw [black] (8,0) circle (2pt) {node[above]{}};
\filldraw [black] (10,0) circle (2pt) {node[above]{}};
\filldraw [black] (6,2) circle (2pt) {node[below]{}};
\filldraw [black] (8,2) circle (2pt) {node[above]{}};
\filldraw [black] (10,2) circle (2pt) {node[above]{}};
\draw[-<] (6,0) -- (6,1.5);
\draw[-] (6,1.5) -- (6,2);
\draw[->] (6,0) -- (7.5,1.5);
\draw[-] (7.5,1.5) -- (8,2);
\draw[->] (6,0) -- (9,1.5);
\draw[-] (9,1.5) -- (10,2);
\draw[->] (8,0) -- (6.5,1.5);
\draw[-] (6.5,1.5) -- (6,2);
\draw[-<] (8,0) -- (8,1.5);
\draw[-] (8,1.5) -- (8,2);
\draw[->] (8,0) -- (9.5,1.5);
\draw[-] (9.5,1.5) -- (10,2);
\draw[->] (10,0) -- (7,1.5);
\draw[-] (7,1.5) -- (6,2);
\draw[->] (10,0) -- (8.5,1.5);
\draw[-] (8.5,1.5) -- (8,2);
\end{tikzpicture}
}}

\caption{A cubic oriented graph with pushable chromatic number~$6$ (i),
and an  oriented graph with maximum average degree strictly less than $3$ and pushable chromatic number~$5$ (ii).}~\label{orientable coloring girth 4}
\end{figure}

The lower bound follows from the existence of subcubic oriented graphs with pushable chromatic number~$6$,
such as the one depicted in Figure~\ref{orientable coloring girth 4}(i).
To prove the upper bound of Theorem~\ref{th subcubic} we show that any subcubic oriented graph 
$\overrightarrow{G}$ admits a pushable homomorphism to the Paley tournament $\overrightarrow{{\rm Pal}_7}$ on seven vertices.
We prove this is the rest of this section.

Assume that this does not hold for all subcubic oriented graphs,
and consider $\overrightarrow{H}$ a minimum (with respect to its number of vertices) subcubic oriented graph that does not admit a pushable homomorphism to $\overrightarrow{{\rm Pal}_7}$.   
We prove that  
 $\overrightarrow{H}$ cannot contain 
 certain  configurations until we finally reach a contradiction to its existence.
 Note that $\overrightarrow{H}$ must be connected due to the minimality condition. 
 
We first show that  $\overrightarrow{H}$ must be cubic.  

 \begin{lemma}
$\overrightarrow{H}$ is cubic. 
 \end{lemma}
 
 \begin{proof}
 Assume $\overrightarrow{H}$ has a degree-$1$ vertex $u$.
By minimality, there exists a pushable homomorphism $f$ from $\overrightarrow{H} - \{u\}$ to  $\overrightarrow{{\rm Pal}_7}$. It is possible to extend $f$ to a pushable homomorphism from 
  $\overrightarrow{H}$
to $\overrightarrow{{\rm Pal}_7}$ due to Property $P(1,6)$, a contradiction.
Now assume $\overrightarrow{H}$ has a degree-$2$ vertex $u$ with neighbors $v$ and $w$.
Consider the oriented graph $\overrightarrow{H_1}$ obtained from $\overrightarrow{H}$ 
by deleting $u$  and adding 
the arc $wv$  if $v$ and $w$ are not already adjacent. 
 By minimality, there exists a pushable homomorphism $f$ from $\overrightarrow{H_1}$ to $\overrightarrow{{\rm Pal}_7}$.
 It is possible to extend $f$ to a pushable homomorphism from 
  $\overrightarrow{H}$
to $\overrightarrow{{\rm Pal}_7}$ due to Property $P(2,2)$, a contradiction.
 \end{proof}
 
 Also, we note that $\overrightarrow{H}$ cannot be a tournament. 
 
 \begin{lemma}
 $\overrightarrow{H}$ is not a tournament. 
 \end{lemma}
 
 \begin{proof}
 Any tournament on four vertices is in a push relationship with one of the following 
 two tournaments (both contained in $\overrightarrow{{\rm Pal}_7}$): $(i)$ the induced tournament $\overrightarrow{{\rm Pal}_7}[\{0,1,2,4\}]$  and  
 $(ii)$ the induced tournament $\overrightarrow{{\rm Pal}_7}[\{1,2,3,4\}]$.
 Thus, if $\overrightarrow{H}$ is an orientation of $K_4$ then it must admit a pushable homomorphism to $\overrightarrow{{\rm Pal}_7}$, a contradiction.
 \end{proof}
 
  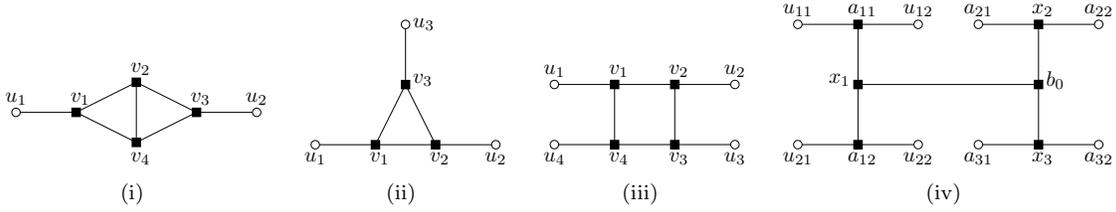
\begin{figure} 
\centering
\subfloat[]{
\scalebox{0.8}{
\begin{tikzpicture}
\draw[-] (0,0) -- (1,0);
\draw[-] (2,.5) -- (1,0);
\draw[-] (2,-.5) -- (1,0);
\draw[-] (2,.5) -- (3,0);
\draw[-] (2,-.5) -- (3,0);
\draw[-] (3,0) -- (4,0);
\draw[-] (2,-.5) -- (2,.5);
\filldraw [black,draw,fill=white] (0,0) circle (2pt) {node[above]{$u_1$}};
\filldraw [black] ([xshift=-2pt,yshift=-2pt]1,0) rectangle ++(4pt,4pt) {node[above,yshift=-2pt]{$v_1$}};
\filldraw [black] ([xshift=-2pt,yshift=-2pt]2,.5) rectangle ++(4pt,4pt) {node[above,yshift=-2pt]{$v_2$}};
\filldraw [black] ([xshift=-2pt,yshift=-2pt]2,-.5) rectangle ++(4pt,4pt) {node[below,yshift=-3pt]{$v_4$}};	
\filldraw [black] ([xshift=-2pt,yshift=-2pt]3,0) rectangle ++(4pt,4pt) {node[above,yshift=-2pt]{$v_3$}};
\filldraw [black,draw,fill=white] (4,0) circle (2pt) {node[above]{$u_2$}};
\end{tikzpicture}
}}
\subfloat[]{
\scalebox{0.8}{
\begin{tikzpicture}
\draw[-] (5,-.5) -- (8,-.5);
\draw[-] (6,-.5) -- (6.5,.5);
\draw[-] (7,-.5) -- (6.5,.5);
\draw[-] (6.5,1.5) -- (6.5,.5);
\filldraw [black,draw,fill=white] (5,-.5) circle (2pt) {node[below]{$u_1$}};
\filldraw [black] ([xshift=-2pt,yshift=-2pt]6,-.5) rectangle ++(4pt,4pt) {node[below,yshift=-2pt]{$v_1$}};
\filldraw [black] ([xshift=-2pt,yshift=-2pt]7,-.5) rectangle ++(4pt,4pt) {node[below,yshift=-2pt]{$v_2$}};
\filldraw [black,draw,fill=white] (8,-.5) circle (2pt) {node[below]{$u_2$}};
\filldraw [black] ([xshift=-2pt,yshift=-2pt]6.5,.5) rectangle ++(4pt,4pt) {node[right,xshift=-2]{$v_3$}};
\filldraw [black,draw,fill=white] (6.5,1.5) circle (2pt) {node[right]{$u_3$}};
\end{tikzpicture}
}}
\subfloat[]{
\scalebox{0.8}{
\begin{tikzpicture}
\draw[-] (0.5,-2) -- (3.5,-2);
\draw[-] (0.5,-3) -- (3.5,-3);
\draw[-] (1.5,-2) -- (1.5,-3);
\draw[-] (2.5,-2) -- (2.5,-3);
\filldraw [black,draw,fill=white] (0.5,-2) circle (2pt) {node[above]{$u_1$}};
\filldraw [black] ([xshift=-2pt,yshift=-2pt]1.5,-2) rectangle ++(4pt,4pt) {node[above,yshift=-2pt]{$v_1$}};
\filldraw [black] ([xshift=-2pt,yshift=-2pt]2.5,-2) rectangle ++(4pt,4pt) {node[above,yshift=-2pt]{$v_2$}};
\filldraw [black,draw,fill=white] (3.5,-2) circle (2pt) {node[above]{$u_2$}};
\filldraw [black,draw,fill=white] (0.5,-3) circle (2pt) {node[below]{$u_4$}};
\filldraw [black] ([xshift=-2pt,yshift=-2pt]1.5,-3) rectangle ++(4pt,4pt) {node[below,yshift=-2pt]{$v_4$}};
\filldraw [black] ([xshift=-2pt,yshift=-2pt]2.5,-3) rectangle ++(4pt,4pt) {node[below,yshift=-2pt]{$v_3$}};
\filldraw [black,draw,fill=white] (3.5,-3) circle (2pt) {node[below]{$u_3$}};
\end{tikzpicture}
}}
\subfloat[]{
\scalebox{0.8}{
\begin{tikzpicture}
\draw[-] (4,-1.5) -- (6,-1.5);
\draw[-] (4,-3.5) -- (6,-3.5);
\draw[-] (7,-1.5) -- (9,-1.5);
\draw[-] (7,-3.5) -- (9,-3.5);
\draw[-] (5,-1.5) -- (5,-3.5);
\draw[-] (8,-1.5) -- (8,-3.5);
\draw[-] (5,-2.5) -- (8,-2.5);
\filldraw [black,draw,fill=white] (4,-1.5) circle (2pt) {node[above]{$u_{11}$}};
\filldraw [black] ([xshift=-2pt,yshift=-2pt]5,-1.5) rectangle ++(4pt,4pt) {node[above,yshift=-2pt]{$a_{11}$}};
\filldraw [black,draw,fill=white] (6,-1.5) circle (2pt) {node[above]{$u_{12}$}};
\filldraw [black] ([xshift=-2pt,yshift=-2pt]5,-2.5) rectangle ++(4pt,4pt) {node[left,xshift=-2pt]{$x_1$}};
\filldraw [black,draw,fill=white] (4,-3.5) circle (2pt) {node[below]{$u_{21}$}};
\filldraw [black] ([xshift=-2pt,yshift=-2pt]5,-3.5) rectangle ++(4pt,4pt) {node[below,yshift=-2pt]{$a_{12}$}};
\filldraw [black,draw,fill=white] (6,-3.5) circle (2pt) {node[below]{$u_{22}$}};
\filldraw [black,draw,fill=white] (7,-1.5) circle (2pt) {node[above]{$a_{21}$}};
\filldraw [black] ([xshift=-2pt,yshift=-2pt]8,-1.5) rectangle ++(4pt,4pt) {node[above,yshift=-2pt]{$x_2$}};
\filldraw [black,draw,fill=white] (9,-1.5) circle (2pt) {node[above]{$a_{22}$}};
\filldraw [black] ([xshift=-2pt,yshift=-2pt]8,-2.5) rectangle ++(4pt,4pt) {node[right,xshift=-2pt]{$b_0$}};
\filldraw [black,draw,fill=white] (7,-3.5) circle (2pt) {node[below]{$a_{31}$}};
\filldraw [black] ([xshift=-2pt,yshift=-2pt]8,-3.5) rectangle ++(4pt,4pt) {node[below,yshift=-2pt]{$x_3$}};
\filldraw [black,draw,fill=white] (9,-3.5) circle (2pt) {node[below]{$a_{32}$}};
\end{tikzpicture}
}}
\caption{Configurations needed for proving Theorem~\ref{th subcubic}. Black square vertices are vertices whose full neighborhood is part of the configuration. 
White circle vertices are vertices that might have other neighbors outside the configuration.}~\label{fig subcubic}
\end{figure}	
 
 Observe that a connected cubic oriented graph that is not a tournament (i.e., not an orientation of~$K_4$),  must have one of the configurations depicted in Figure~\ref{fig subcubic}. 
In what follows, we prove that  none of these configurations 
 can be present in $\overrightarrow{H}$, and thus that it cannot exist.
 
 Before going on to show that $\overrightarrow{H}$ does not contain the said configurations, we first introduce some notation and raise some remarks. We below deal with \emph{partial matrices}, i.e. matrices whose entries are either empty or contain an element of $V(\overrightarrow{{\rm Pal}_7})$. The $ij^{th}$ entry of a matrix $X$ is denoted by $X(i,j)$. Given two matrices $X_1$ and $X_2$ of the same dimension, the matrices $X_1 \pm X_2$ are well defined by setting $(X_1 \pm X_2)(i,j) = X_1(i,j) \pm X_2(i,j)$, with the convention that $\emptyset\pm x=x\pm\emptyset=\emptyset$ for every entry $x$.

In the upcoming lemmas, we will often use implicitely the following observation to check the correctness of some extensions of pushable homomorphisms : if, for some $(i,j)$, we have $(X_2 - X_1)(i,j) \in \{1,2,4\}$, then taking $f(u) = X_1(i,j)$ and $f(v) = X_2(i,j)$ defines a homomorphism of the arc $uv$ to 
$\overrightarrow{{\rm Pal}_7}$. Similarly, if $(X_2 - X_1)(i,j) \in \{3,5,6\}$ for some $(i,j)$, then taking $f(u) = X_1(i,j)$ and $f(v) = X_2(i,j)$ defines a homomorphism of the arc $vu$ to $\overrightarrow{{\rm Pal}_7}$. 

 \begin{lemma}\label{lem config i}
 The configuration depicted in Figure~\ref{fig subcubic}(i) cannot be contained  in $\overrightarrow{H}$.
 \end{lemma}
 
 \begin{proof}
 Assume that $\overrightarrow{H}$ contains the configuration depicted in Figure~\ref{fig subcubic}(i). 
 Let $\overrightarrow{H_1}$ be the oriented graph obtained from $\overrightarrow{H}$ by deleting the vertices in
 $\{v_1, v_2, v_3, v_4\}$, and let $\overrightarrow{H_2}$ be the oriented graph obtained from 
 $\overrightarrow{H}$ by deleting the arc between $u_2$ and $v_3$. 
By minimality, $\overrightarrow{H_1}$ admits a pushable homomorphism $f$ to 
$\overrightarrow{{\rm Pal}_7}$. Up to pushing vertices in $\overrightarrow{H_1}$, we may assume that $f$ is actually an oriented homomorphism. As $\overrightarrow{{\rm Pal}_7}$ is vertex-transitive, without loss of generality we may assume that $f(u_1) = 0$. Moreover, up to pushing $v_1, v_2$ and $v_4$ in that order, we may assume that $\overrightarrow{H}$ has the arcs $u_1v_1, v_2v_1, v_4v_1$. Furthermore, up to exchanging $v_2$ and $v_4$ and then pushing $v_3$, we may also assume that $\overrightarrow{H_2}$ has the arcs $v_2v_4$ and $v_3v_2$. 

We show that for every $\ell\in\{1,\ldots,6\}$ we can extend $f$ to an oriented homomorphism from $\overrightarrow{H_2}$ to $\overrightarrow{{\rm Pal}_7}$ satisfying $f(v_3)=\ell$. This allows to conclude: let $\beta\in\{+,-\}$ such that $v_3$ is a $\beta$-neighbor of $u_2$ in $\overrightarrow{H}$ after having possibly pushed $v_3$ to obtain the arc $v_3v_2$. Since $|N^{\beta}(f(u_2))| = 3$, there exists $\ell \in N^{\beta}(f(u_2)) \setminus \{0\}$. We then extend $f$ to an oriented homomorphism from 
$\overrightarrow{H_2}$ to 
$\overrightarrow{{\rm Pal}_7}$ such that $f(v_3)=\ell$. Due to the choice of $\ell$, this is also an oriented homomorphism from $\overrightarrow{H}$ to 
$\overrightarrow{{\rm Pal}_7}$, a contradiction.

Therefore, in order to prove that $\overrightarrow{H}$ cannot contain the configuration depicted in 
Figure~\ref{fig subcubic}(i), it only remains to extend $f$ to an oriented homomorphism from $\overrightarrow{H_2}$ to $\overrightarrow{{\rm Pal}_7}$ satisfying $f(v_3)=\ell$ for every $\ell\in\{1,\ldots,6\}$. To this end, we consider the following matrices:  
$$
X_{v_1} = 
\begin{bmatrix}
1 & 2 & 4\\
1 & 2 & 4\\
1 & 2 & 4
\end{bmatrix},~ 
X_{v_2} =  
\begin{bmatrix}
0 & 1 & 3\\
6 & 0 & 2\\
4 & 5 & 0
\end{bmatrix},~ 
X_{v_4} =  
\begin{bmatrix}
4 & 5 & 0\\
0 & 1 & 3\\
6 & 0 & 2
\end{bmatrix},
$$
and
$$
X^+_{v_3} =  
\begin{bmatrix}
6 & 0 & 2\\
4 & 5 & 0\\
0 & 1 & 3
\end{bmatrix},~
X^-_{v_3} =  
\begin{bmatrix}
3 & 4 & 6\\
5 & 6 & 1\\
2 & 3 & 5
\end{bmatrix}.
$$

Let $\alpha\in \{+,-\}$ such that $v_3 \in N^{\alpha}(v_4)$, and $\ell\in\{1,\ldots,6\}$. Observe that all values $\{1,\ldots,6\}$ are present in the matrix $X^{\alpha}_{v_3}$, hence we can take $(i,j)\in\{1,2,3\}^2$ such that $\ell= X^{\alpha}_{v_3}(i,j)$. We can then extend $f$ to an oriented homomorphism from $\overrightarrow{H_2}$ to $\overrightarrow{{\rm Pal}_7}$ by
choosing $f(v_k) = X_{v_k}(i,j)$ for all $k \in \{1,2,4\}$ and 
$f(v_3) =X^{\alpha}_{v_3}(i,j)=\ell $, which concludes the proof.
\end{proof}

 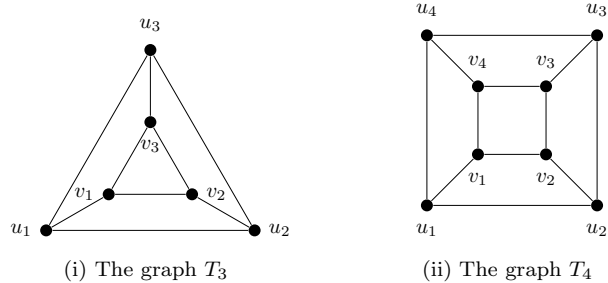
\begin{figure} 
\centering
\subfloat[The graph $T_3$]{
\scalebox{0.8}{
\begin{tikzpicture}[every node/.style = {circle, fill=black,inner sep = 2pt}]
\node[label = below:{$v_3$}] (v3) at (90:0.8){};
\node[label = above:{$u_3$}] (u3) at (90:2){};
\node[label = right:{$v_2$}] (v2) at (-30:0.8){};
\node[label = right:{$u_2$}] (u2) at (-30:2){};
\node[label = left:{$v_1$}] (v1) at (210:0.8){};
\node[label = left:{$u_1$}] (u1) at (210:2){};
\draw (u1) -- (u3) -- (u2) -- (u1) -- (v1) -- (v2) -- (v3) -- (v1) ;
\draw (u3) -- (v3);
\draw (u2) -- (v2);
\end{tikzpicture}
}}
\hspace{30pt}
\subfloat[The graph $T_4$]{
\scalebox{0.8}{
\begin{tikzpicture}[every node/.style = {circle, fill=black,inner sep = 2pt}]
\node[label = above:{$u_4$}] (u4) at (135:2) {};
\node[label = above:{$v_4$}] (v4) at (135:0.8) {};
\node[label = above:{$u_3$}] (u3) at (45:2) {};
\node[label = above:{$v_3$}] (v3) at (45:0.8) {};
\node[label = below:{$u_2$}] (u2) at (-45:2) {};
\node[label = below:{$v_2$}] (v2) at (-45:0.8) {};
\node[label = below:{$u_1$}] (u1) at (-135:2) {};
\node[label = below:{$v_1$}] (v1) at (-135:0.8) {};
\draw (u1) -- (u2) -- (u3) -- (u4) -- (u1) -- (v1) -- (v2) -- (v3) -- (v4) -- (v1);
\draw (u2) -- (v2);
\draw (u3) -- (v3);
\draw (u4) -- (v4);
\end{tikzpicture}
}}
\caption{Two cubic graphs mentioned in the proof of Theorem~\ref{th subcubic}.}~\label{fig support}
\end{figure}

Before moving on to proving Lemma~\ref{lem config ii}, we need to show the following.

\begin{lemma}\label{lem support i}
The graph $T_3$ depicted in Figure~\ref{fig support}(i) cannot be contained in $H$.  
 \end{lemma}
 
 \begin{proof}
Since $T_3$ is cubic and $H$ is connected, if $H$ contains $T_3$ then $H=T_3$. Therefore, it is enough to show that for any orientation $\overrightarrow{T_3}$ of $T_3$, there exists a pushable homomorphism $f$ from $\overrightarrow{T_3}$ to $\overrightarrow{{\rm Pal}_7}$.
Note that regardless of the orientation $\overrightarrow{T_3}$, it is always possible to push some vertices among $\{u_1,u_2,u_3\}$ so that $\overrightarrow{T_3}$ has the arcs $u_1u_2, u_2u_3, u_3u_1$. Moreover, we can also push some of the vertices among $\{v_1, v_2, v_3\}$ to obtain the arcs $u_1v_1, u_2v_2, u_3v_3$ as well. 

We now define a homomorphism from $\overrightarrow{T_3}$ to $\overrightarrow{{\rm Pal}_7}$. We first set $f(u_1) = 1$, $f(u_2) = 2$ and $f(u_3) = 4$. As shown in Figure~\ref{fig:T3}, whatever the orientation of the triangle induced by  $\{v_1, v_2, v_3\}$ in $\overrightarrow{T_3}$ is, we are always able to choose $f(v_1) \in \{2,3,5\}$, $f(v_2) \in \{3,4,6\}$ and  $f(v_3) \in \{1,5,6\}$ to extend $f$ to an oriented homomorphism from $\overrightarrow{T_3}$ to $\overrightarrow{{\rm Pal}_7}$.\qedhere

\begin{figure}[!t]
\centering
 \begin{tikzpicture}[every node/.style={draw,circle,inner sep =1pt},thick,scale=0.8]
     \node[label=above:{$v_1$}] (1) at (210:1) {3};
     \node[label=above:{$v_2$}] (2) at (-30:1) {6};
     \node[label=right:{$v_3$}] (3) at (90:1) {5};
     \draw[->] (1) to (3);
     \draw[->] (3) to (2);
     \draw[->] (2) to (1);
     \tikzset{xshift=3cm}
     \node[label=above:{$v_1$}] (1) at (210:1) {2};
     \node[label=above:{$v_2$}] (2) at (-30:1) {3};
     \node[label=right:{$v_3$}] (3) at (90:1) {6};
     \draw[->] (1) to (2);
     \draw[->] (2) to (3);
     \draw[->] (1) to (3);
     \tikzset{xshift=3cm}
     \node[label=above:{$v_1$}] (1) at (210:1) {5};
     \node[label=above:{$v_2$}] (2) at (-30:1) {4};
     \node[label=right:{$v_3$}] (3) at (90:1) {6};
     \draw[->] (2) to (1);
     \draw[->] (1) to (3);
     \draw[->] (2) to (3);
     \tikzset{xshift=3cm}
     \node[label=above:{$v_1$}] (1) at (210:1) {3};
     \node[label=above:{$v_2$}] (2) at (-30:1) {4};
     \node[label=right:{$v_3$}] (3) at (90:1) {5};
     \draw[->] (1) to (2);
     \draw[->] (2) to (3);
     \draw[->] (1) to (3);
     \tikzset{yshift=-2.5cm,xshift=-9cm}
     \node[label=above:{$v_1$}] (1) at (210:1) {5};
     \node[label=above:{$v_2$}] (2) at (-30:1) {3};
     \node[label=right:{$v_3$}] (3) at (90:1) {1};
     \draw[->] (3) to (2);
     \draw[->] (2) to (1);
     \draw[->] (3) to (1);
     \tikzset{xshift=3cm}
     \node[label=above:{$v_1$}] (1) at (210:1) {2};
     \node[label=above:{$v_2$}] (2) at (-30:1) {3};
     \node[label=right:{$v_3$}] (3) at (90:1) {1};
     \draw[->] (3) to (1);
     \draw[->] (1) to (2);
     \draw[->] (3) to (2);
     \tikzset{xshift=3cm}
     \node[label=above:{$v_1$}] (1) at (210:1) {3};
     \node[label=above:{$v_2$}] (2) at (-30:1) {6};
     \node[label=right:{$v_3$}] (3) at (90:1) {1};
     \draw[->] (2) to (3);
     \draw[->] (3) to (1);
     \draw[->] (2) to (1);
     \tikzset{xshift=3cm}
     \node[label=above:{$v_1$}] (1) at (210:1) {2};
     \node[label=above:{$v_2$}] (2) at (-30:1) {3};
     \node[label=right:{$v_3$}] (3) at (90:1) {5};
     \draw[->] (3) to (1);
     \draw[->] (1) to (2);
     \draw[->] (2) to (3);
 \end{tikzpicture}
 \caption{The 8 orientations of the inner cycle induced by $\{v_1,v_2,v_3\}$.}
 \label{fig:T3}
 \end{figure}
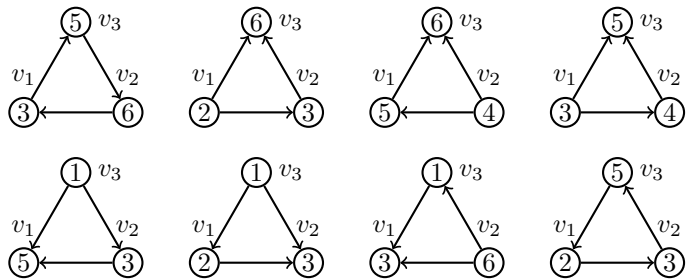
\end{proof}

 \begin{lemma}\label{lem config ii}
 The configuration depicted in Figure~\ref{fig subcubic}(ii) cannot be contained in $\overrightarrow{H}$.
 \end{lemma}
 
 \begin{proof}
  Assume that $\overrightarrow{H}$ contains the configuration depicted in Figure~\ref{fig subcubic}(ii). Due to Lemmas~\ref{lem config i} and~\ref{lem support i}, we may assume that $u_1$ and $u_2$ are distinct non-adjacent vertices. Moreover, it is possible to push some of the vertices among $\{v_1,v_2\}$ to make sure that $\overrightarrow{H}$ has the arcs $u_1v_1, u_2v_2$. Furthermore, by symmetry, we may assume the arc $v_1v_2$ is present in $\overrightarrow{H}$. 
  
  Let $\overrightarrow{H_1}$ be the oriented graph obtained by adding the arc $u_1u_2$ in $\overrightarrow{H}$. We also denote by $\overrightarrow{H_2}$ the oriented graph obtained from $\overrightarrow{H_1}$ by deleting the vertices $v_1, v_2$ and $v_3$. By minimality, $\overrightarrow{H_2}$ admits a pushable homomorphism $f$ to $\overrightarrow{{\rm Pal}_7}$. Up to replacing $\overrightarrow{H_2}$ (together with $\overrightarrow{H_1}$ and $\overrightarrow{H}$) by a push-equivalent oriented graph, we may assume that $f$ is an oriented homomorphism. However, note that this may cause the arc $u_1u_2$ to be reversed in $\overrightarrow{H_2}$. This occurs if we needed to push $u_1$ or $u_2$ in $\overrightarrow{H_2}$ in order to make $f$ an oriented homomorphism. We again push (if needed) $v_1$ and $v_2$ to obtain the arcs $u_1v_1$ and $u_2v_2$ in $\overrightarrow{H}$. Observe that $u_1u_2$ and $v_1v_2$ are both present in $\overrightarrow{H}$ or both reversed. By symmetry, we may only consider the first case. Finally, up to pushing $v_3$, we assume that the arc $v_3v_1$ is in $\overrightarrow{H}$.
  
Let $\overrightarrow{H_3}$ be the oriented graph  obtained from $\overrightarrow{H}$ by deleting the arc between $u_3$ and $v_3$. Similarly to the proof of Lemma~\ref{lem config i}, we first extend $f$ to an oriented homomorphism from $\overrightarrow{H_3}$ to $\overrightarrow{{\rm Pal}_7}$, with some additional constraint on $f(v_3)$, and then extend $f$ to $\overrightarrow{H}$.
  
As $\overrightarrow{{\rm Pal}_7}$ is arc-transitive, without loss of generality we may assume that $f(u_1) = 0$ and $f(u_2)=1$. Now consider the following matrices:  
$$
X_{v_1} = 
\begin{bmatrix}
1 \\
1 \\
1 \\
2 \\
4 
\end{bmatrix},~
X_{v_2} =  
\begin{bmatrix}
2 \\
3 \\
5 \\
3 \\
5 
\end{bmatrix},~
X^+_{v_3} =  
\begin{bmatrix}
4\\0\\6\\5\\2
\end{bmatrix},~
X^-_{v_3} =  
\begin{bmatrix}
0 \\
6 \\
4 \\
1 \\
3 
\end{bmatrix}.
$$

Let $\alpha,\beta\in\{+,-\}$ such that $v_3$ is an $\alpha$-neighbor of $v_2$ and a $\beta$-neighbor of $u_3$. Let $S_\alpha$ be the set of all integers appearing in at least one entry in $X_{v_3}^\alpha$. Observe that for every $\ell \in S_\alpha$, there exists $j\in\{1,\ldots,5\}$ such that $\ell=X^\alpha_{v_3}(1,j)$. By choosing $f(v_k) = X_{v_k}(1,j)$ for $k \in \{1,2\}$ and $f(v_3)=\ell$, we can extend $f$ to an oriented homomorphism from $\overrightarrow{H_3}$ to $\overrightarrow{{\rm Pal}_7}$ satisfying $f(v_3)=\ell$.

Observe that $|S_\alpha|=5$. Hence, since $|N^{\beta}(f(u_2))| = 3$, we may choose an $\ell \in N^{\beta}(f(u_2)) \cap S_\alpha$. The corresponding extension of $f$ is now an oriented homomorphism from 
$\overrightarrow{H_1}$, to $\overrightarrow{{\rm Pal}_7}$. Since $\overrightarrow{H}$ is a subgraph of $\overrightarrow{H_1}$, we obtain a contradiction. Therefore, $\overrightarrow{H}$ cannot  contain the configuration depicted in Figure~\ref{fig subcubic}(ii). 
\end{proof}

The proof of Lemma~\ref{lem config iii} is similar to Lemma~\ref{lem config ii}. In particular, we first prove an auxiliary lemma in the spirit of Lemma~\ref{lem support i}.

\begin{lemma}\label{lem support ii}
The graph $H$ cannot be the graph $T_4$ depicted in Figure~\ref{fig support}(ii).  
 \end{lemma}
 
 \begin{proof}
It is enough to show that for any orientation $\overrightarrow{T_4}$ of the cubic graph $T_4$ depicted in 
Figure~\ref{fig support}(ii) there exists a pushable homomorphism $f$ from $\overrightarrow{T_4}$ to 
$\overrightarrow{{\rm Pal}_7}$. We consider two cases, depending on parity of the number of backward arcs of the 4-cycles of $\overrightarrow{T_4}$. 

\medskip
\noindent \textit{Case 1:} Suppose that $\overrightarrow{T_4}$ contains a 4-cycle with an odd number of backward arcs. By symmetry, assume that this cycle is the outer one. We can now push some vertices among $\{u_1,u_2,u_3,u_4\}$ to make sure that $\overrightarrow{T_4}$ has the arcs $u_1u_2, u_1u_4, u_4u_3$ and $u_3u_2$. Up to pushing some vertices among $\{v_1,v_2,v_3,v_4\}$, we may also assume that $\overrightarrow{T_4}$ contains all the arcs $u_iv_i$ for $i\in\{1,2,3,4\}$. Let $f(u_1) = 0, f(u_2) = 4, f(u_3) = 2$ and $f(u_4) = 1$. As shown in Figure~\ref{fig:T4}, whatever the orientation of the $4$-cycle induced by  $\{v_1, v_2, v_3, v_4\}$ is, we are  always able to choose $f(v_1) \in \{1,2,4\}$, $f(v_2) \in \{1,5,6\}$,  $f(v_3) \in \{3,4,6\}$ and  $f(v_4) \in \{2,3,5\}$ to extend $f$ to a homomorphism from $\overrightarrow{T_4}$ to  $\overrightarrow{{\rm Pal}_7}$.

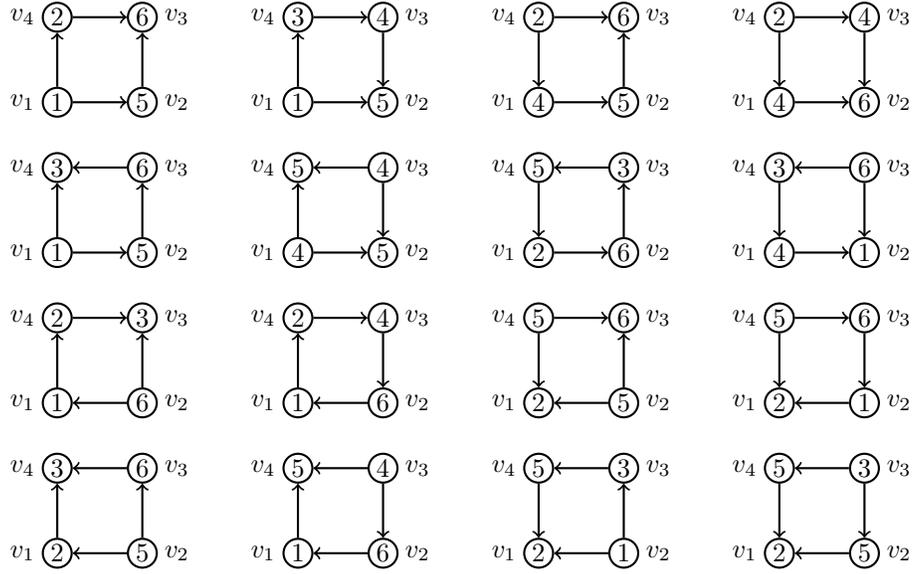
\begin{figure}[!t]
\centering
\begin{tikzpicture}[every node/.style={draw, circle, inner sep=1pt}, thick, scale=0.8]
\node[label=left:{$v_1$}] (1) at (-135:1) {1};
\node[label=right:{$v_2$}] (2) at (-45:1) {5};
\node[label=right:{$v_3$}] (3) at (45:1) {6};
\node[label=left:{$v_4$}] (4) at (135:1) {2};
\draw[->] (1) to (2);
\draw[->] (2) to (3);
\draw[->] (1) to (4);
\draw[->] (4) to (3);
\tikzset{xshift=4cm}
\node[label=left:{$v_1$}] (1) at (-135:1) {1};
\node[label=right:{$v_2$}] (2) at (-45:1) {5};
\node[label=right:{$v_3$}] (3) at (45:1) {4};
\node[label=left:{$v_4$}] (4) at (135:1) {3};
\draw[->] (1) to (2);
\draw[->] (3) to (2);
\draw[->] (1) to (4);
\draw[->] (4) to (3);
\tikzset{xshift=4cm}
\node[label=left:{$v_1$}] (1) at (-135:1) {4};
\node[label=right:{$v_2$}] (2) at (-45:1) {5};
\node[label=right:{$v_3$}] (3) at (45:1) {6};
\node[label=left:{$v_4$}] (4) at (135:1) {2};
\draw[->] (1) to (2);
\draw[->] (2) to (3);
\draw[->] (4) to (1);
\draw[->] (4) to (3);
\tikzset{xshift=4cm}
\node[label=left:{$v_1$}] (1) at (-135:1) {4};
\node[label=right:{$v_2$}] (2) at (-45:1) {6};
\node[label=right:{$v_3$}] (3) at (45:1) {4};
\node[label=left:{$v_4$}] (4) at (135:1) {2};
\draw[->] (1) to (2);
\draw[->] (3) to (2);
\draw[->] (4) to (1);
\draw[->] (4) to (3);
\tikzset{xshift=-12cm,yshift=-2.5cm}
\node[label=left:{$v_1$}] (1) at (-135:1) {1};
\node[label=right:{$v_2$}] (2) at (-45:1) {5};
\node[label=right:{$v_3$}] (3) at (45:1) {6};
\node[label=left:{$v_4$}] (4) at (135:1) {3};
\draw[->] (1) to (2);
\draw[->] (2) to (3);
\draw[->] (1) to (4);
\draw[<-] (4) to (3);
\tikzset{xshift=4cm}
\node[label=left:{$v_1$}] (1) at (-135:1) {4};
\node[label=right:{$v_2$}] (2) at (-45:1) {5};
\node[label=right:{$v_3$}] (3) at (45:1) {4};
\node[label=left:{$v_4$}] (4) at (135:1) {5};
\draw[->] (1) to (2);
\draw[->] (3) to (2);
\draw[->] (1) to (4);
\draw[<-] (4) to (3);
\tikzset{xshift=4cm}
\node[label=left:{$v_1$}] (1) at (-135:1) {2};
\node[label=right:{$v_2$}] (2) at (-45:1) {6};
\node[label=right:{$v_3$}] (3) at (45:1) {3};
\node[label=left:{$v_4$}] (4) at (135:1) {5};
\draw[->] (1) to (2);
\draw[->] (2) to (3);
\draw[->] (4) to (1);
\draw[<-] (4) to (3);
\tikzset{xshift=4cm}
\node[label=left:{$v_1$}] (1) at (-135:1) {4};
\node[label=right:{$v_2$}] (2) at (-45:1) {1};
\node[label=right:{$v_3$}] (3) at (45:1) {6};
\node[label=left:{$v_4$}] (4) at (135:1) {3};
\draw[->] (1) to (2);
\draw[->] (3) to (2);
\draw[->] (4) to (1);
\draw[<-] (4) to (3);
\tikzset{xshift=-12cm,yshift=-2.5cm}
\node[label=left:{$v_1$}] (1) at (-135:1) {1};
\node[label=right:{$v_2$}] (2) at (-45:1) {6};
\node[label=right:{$v_3$}] (3) at (45:1) {3};
\node[label=left:{$v_4$}] (4) at (135:1) {2};
\draw[<-] (1) to (2);
\draw[->] (2) to (3);
\draw[->] (1) to (4);
\draw[->] (4) to (3);
\tikzset{xshift=4cm}
\node[label=left:{$v_1$}] (1) at (-135:1) {1};
\node[label=right:{$v_2$}] (2) at (-45:1) {6};
\node[label=right:{$v_3$}] (3) at (45:1) {4};
\node[label=left:{$v_4$}] (4) at (135:1) {2};
\draw[<-] (1) to (2);
\draw[->] (3) to (2);
\draw[->] (1) to (4);
\draw[->] (4) to (3);
\tikzset{xshift=4cm}
\node[label=left:{$v_1$}] (1) at (-135:1) {2};
\node[label=right:{$v_2$}] (2) at (-45:1) {5};
\node[label=right:{$v_3$}] (3) at (45:1) {6};
\node[label=left:{$v_4$}] (4) at (135:1) {5};
\draw[<-] (1) to (2);
\draw[->] (2) to (3);
\draw[->] (4) to (1);
\draw[->] (4) to (3);
\tikzset{xshift=4cm}
\node[label=left:{$v_1$}] (1) at (-135:1) {2};
\node[label=right:{$v_2$}] (2) at (-45:1) {1};
\node[label=right:{$v_3$}] (3) at (45:1) {6};
\node[label=left:{$v_4$}] (4) at (135:1) {5};
\draw[<-] (1) to (2);
\draw[->] (3) to (2);
\draw[->] (4) to (1);
\draw[->] (4) to (3);
\tikzset{xshift=-12cm,yshift=-2.5cm}
\node[label=left:{$v_1$}] (1) at (-135:1) {2};
\node[label=right:{$v_2$}] (2) at (-45:1) {5};
\node[label=right:{$v_3$}] (3) at (45:1) {6};
\node[label=left:{$v_4$}] (4) at (135:1) {3};
\draw[<-] (1) to (2);
\draw[->] (2) to (3);
\draw[->] (1) to (4);
\draw[<-] (4) to (3);
\tikzset{xshift=4cm}
\node[label=left:{$v_1$}] (1) at (-135:1) {1};
\node[label=right:{$v_2$}] (2) at (-45:1) {6};
\node[label=right:{$v_3$}] (3) at (45:1) {4};
\node[label=left:{$v_4$}] (4) at (135:1) {5};
\draw[<-] (1) to (2);
\draw[->] (3) to (2);
\draw[->] (1) to (4);
\draw[<-] (4) to (3);
\tikzset{xshift=4cm}
\node[label=left:{$v_1$}] (1) at (-135:1) {2};
\node[label=right:{$v_2$}] (2) at (-45:1) {1};
\node[label=right:{$v_3$}] (3) at (45:1) {3};
\node[label=left:{$v_4$}] (4) at (135:1) {5};
\draw[<-] (1) to (2);
\draw[->] (2) to (3);
\draw[->] (4) to (1);
\draw[<-] (4) to (3);
\tikzset{xshift=4cm}
\node[label=left:{$v_1$}] (1) at (-135:1) {2};
\node[label=right:{$v_2$}] (2) at (-45:1) {5};
\node[label=right:{$v_3$}] (3) at (45:1) {3};
\node[label=left:{$v_4$}] (4) at (135:1) {5};
\draw[<-] (1) to (2);
\draw[->] (3) to (2);
\draw[->] (4) to (1);
\draw[<-] (4) to (3);
\end{tikzpicture}
\caption{The 16 possible orientations of the inner cycle induced by $\{v_1,v_2,v_3,v_4\}$.}
\label{fig:T4}
\end{figure}

\medskip
\noindent\textit{Case 2:} Suppose that every 4-cycle of $\overrightarrow{T_4}$ has an even number of backward arcs. Up to pushing some vertices among $\{u_1,u_2,u_3,u_4\}$, we may assume that $\overrightarrow{T_4}$ has the arcs $u_1u_2, u_1u_4, u_3u_2$ and $u_3u_4$. We also push some vertices among $\{v_1,v_2,v_3,v_4\}$ such that $\overrightarrow{T_4}$ contains all the arcs $u_iv_i$ for $i\in\{1,2,3,4\}$. By hypothesis, observe that $\overrightarrow{T_4}$ must contain the arcs $v_1v_2,v_1v_4, v_3v_2$ and $v_3v_4$. We can then define an oriented homomorphism from $\overrightarrow{T_4}$ to  $\overrightarrow{{\rm Pal}_7}$ by setting $f(u_1)=f(u_3)=0$, $f(v_1)=f(v_3)=f(u_2)=f(u_4)=1$ and $f(v_2)=f(v_4)=2$.
 \end{proof} 
 
 Using this auxiliary lemma, we may prove that $\overrightarrow{H}$ does not contain the next configuration.
 \begin{lemma}\label{lem config iii}
 The configuration depicted in Figure~\ref{fig subcubic}(iii) cannot be contained in $\overrightarrow{H}$.
 \end{lemma}
 
  \begin{proof}
  Assume that $\overrightarrow{H}$ contains the configuration depicted in Figure~\ref{fig subcubic}(iii). We follow the same approach as in Lemma~\ref{lem config ii}: we first use Lemma~\ref{lem support ii} and symmetry to show we can add the are $u_1u_2$. Then we remove some vertices and use minimality to obtain a pushable homomorphism, that we extend in two steps to a pushable homomorphism from $\overrightarrow{H}$ to $\overrightarrow{{\rm Pal}_7}$.

  We first assume that $u_1$ and $u_2$ are distinct non-adjacent vertices due to Lemmas~\ref{lem config ii} and~\ref{lem support ii}. Up to renaming the vertices of the graph, we assume that if $u_3, u_4$ are adjacent, then either $u_1, u_4$ are adjacent or $u_2, u_3$ are adjacent. Moreover, up to pushing $v_1$ or $v_2$, we may assume that $\overrightarrow{H}$ has the arcs $v_1u_1, v_2u_2$. Furthermore, by symmetry, we may assume the arc $v_2v_1$ is present in $\overrightarrow{H}$. 

  Let $\overrightarrow{H_1}$ be the oriented graph obtained from $\overrightarrow{H}$ by adding the arc $u_1u_2$ and an arc between $u_3$ and $u_4$ (if such is not already present) in such a way that the oriented cycle induced by the vertices $\{u_3, v_3, v_4, u_4\}$ has an even number of forward and backward arcs.  

  After that, let $\overrightarrow{H_2}$ be the oriented graph obtained from $\overrightarrow{H_1}$ by deleting the set of vertices $\{v_1, v_2, v_3, v_4\}$, and $\overrightarrow{H_3}$ be the oriented graph obtained from $\overrightarrow{H}$ by deleting the arc between $u_3$ and $v_3$, and the arc between $u_4$ and $v_4$.
  
  By minimality, $\overrightarrow{H_2}$ admits a pushable homomorphism $f$ to $\overrightarrow{{\rm Pal}_7}$. Again, up to replacing $\overrightarrow{H_2}$ (together with $\overrightarrow{H_1}$ and $\overrightarrow{H}$) by a push-equivalent oriented graph, we may assume that $f$ is an oriented homomorphism.  Observe that the arc $u_1u_2$ may now be reversed. We push $v_1$ or $v_2$ if needed to make sure that $\overrightarrow{H_2}$ contains $v_1u_1$ and $v_2u_2$. Now either $\overrightarrow{H_2}$ contains both $u_1u_2$ and $v_2v_1$ or both $u_2u_1$ and $v_1v_2$. By symmetry, we consider only the first case. Moreover, since $\overrightarrow{{\rm Pal}_7}$ is arc-transitive, we may assume that $f(u_1)=0$ and $f(u_2)=1$.
  
  Up to pushing $v_3$ or $v_4$ if needed, we may assume that $\overrightarrow{H_1}$ contains the arcs $v_1v_4$ and $v_2v_3$. Now consider the following matrices:  

 $$
 X_{v_1} = 
\left[
\begin{array}{c c c|c c c}
3 & 5 & 6 & 3 & 5 & 6 \\
3 & 5 & 6 & 3 & 5 & 6 \\
3 & 5 & 6 & 3 & 5 & 6 \\
\hline
2 & 2 & 4 & 2 & 4 & 4 \\
2 & 2 & 4 & 2 & 4 & 4 \\
2 & 2 & 4 & 2 & 4 & 4 
\end{array}
\right],~
 X_{v_2} = 
\left[
\begin{array}{c c c |c c c}
6 & 4 & 4 & 5 & 2 & 3 \\
6 & 4 & 4 & 5 & 2 & 3 \\
6 & 4 & 4 & 5 & 2 & 3 \\
\hline
4 & 6 & 6 & 5 & 2 & 3 \\
4 & 6 & 6 & 5 & 2 & 3 \\
4 & 6 & 6 & 5 & 2 & 3 
\end{array}
\right],~
X_{v_4} =  
\left[
\begin{array}{c c c |c c c}
4 & 6 & 0 & 4 & 6 & 0 \\
5 & 0 & 1 & 5 & 0 & 1 \\
0 & 2 & 3 & 0 & 2 & 3 \\
\hline
5 & 5 & 0 & 5 & 0 & 0 \\
0 & 0 & 2 & 0 & 2 & 2 \\
1 & 1 & 3 & 1 & 3 & 3 \\
\end{array}
\right],
$$

  $$
 X^+_{v_3} = 
\left[
\begin{array}{c c c |c c c}
1  & 1 & 1 & 0 & 0 & 2 \\
0 & 1 & 5 & \emptyset & 1 & 2 \\
1 & 6 & 5 & 4 & \emptyset & \emptyset \\
\hline
6 & 0 & 1 & \emptyset & 1 & 2 \\
1 & 1 & 3 & 4 & \emptyset & 6 \\
5 & 3 & 0 & 3 & 5 & \emptyset 
\end{array}
\right],~
 X^-_{v_3} = 
\left[
\begin{array}{c c c |c c c}
0  & 5 & 6 & 3 & 5 & 6 \\
1 & 5 & 6 & 1 & 5 & 6 \\
3 & 5 & 6 & 3 & 5 & 6 \\
\hline
1 & 3 & \emptyset & 4 & 5 & 6 \\
5 & 3 & 1 & 3 & 0 & 1 \\
6 & 0 & 1 & 4 & 1 & 2
\end{array}
\right].
$$



 Let $\overrightarrow{H_3}$ be the oriented graph obtained by deleting the arc between $u_3$ and $v_3$ and the arc between $u_4$ and $v_4$ in $\overrightarrow{H_1}$. For every $\ell\in\{0,\ldots,6\}$, we first extend $f$ to a pushable homomorphism from  $\overrightarrow{H_3}$ to $\overrightarrow{{\rm Pal}_7}$ such that $f(v_3)=\ell$. We then apply this construction with the right value of $\ell$ to obtain a pushable homomorphism from $\overrightarrow{H_1}$ to $\overrightarrow{{\rm Pal}_7}$.

Let $\alpha\in\{+,-\}$ such that $v_3$ is an $\alpha$-neighbor of $v_4$ and let $\ell\in\{0,\ldots,6\}$.
Observe that there exists $i,j\in\{1,\ldots,6\}$ such that $\ell=X_{v_3}^\alpha(i,j)$. We push the vertex $v_1$ (resp. $v_2$) in $\overrightarrow{H_3}$ if $i>3$ (resp. $j>3$). Now, if we choose $f(v_k) = X_{v_k}(i,j)$ for all $k \in \{1,2,4\}$ and any $f(v_3)=\ell$, then $f$ is a pushable homomorphism from $\overrightarrow{H_3}$ to $\overrightarrow{{\rm Pal}_7}$.  

It now remains to choose the value of $\ell$ to apply this construction and extend it to $\overrightarrow{H_1}$. Let $\beta_3,\beta_4\in\{+,-\}$ such that $v_3$ is a $\beta_3$-neighbor of $u_3$, and $v_4$ is a $\beta_4$-neighbor of $u_4$ in $\overrightarrow{H_1}$. Consider the set of integers
$$S^{\alpha}=\left(\bigcup_{\{(i,j) : X_{v_4}(i,j) \in N^{\beta_4}(f(u_4))\}}  \{X^{\alpha}_{v_3}(i,j)\}\right)\setminus\{\emptyset\}.$$
Observe that $S^{\alpha}$ depends on the value of $f(u_4)$ and $\beta_4$. 

Note that for each $\ell \in S^{\alpha}$ there exists an extension of $f$ to a pushable homomorphism from 
$\overrightarrow{H_3}$ to $\overrightarrow{{\rm Pal}_7}$ such that $f(v_4) \in N^{\beta_4}(f(u_4))$ and $f(v_3) = \ell $. In particular, if $\ell$ also lies in $N^{\beta_3}(f(u_3))$, then $f$ is also a pushable homomorphism from $\overrightarrow{H_1}$ to $\overrightarrow{{\rm Pal}_7}$. Therefore, we can conclude as soon as $ N^{\beta_3}(f(u_3))\cap S^\alpha\neq\emptyset$.

This is already the case if $|S^{\alpha}| \geq 5$. Therefore, we may assume that $|S^{\alpha}| < 5$, which may happen for some values of $f(u_4)$ and $\beta_4$. To handle those instances, we split the rest of the proof into two cases. In each of them, we modify $f$ so that $ N^{\beta_3}(f(u_3))\cap S^\alpha\neq\emptyset$, afterwards, which allows to extend $f$ to a pushable homomorphism from $\overrightarrow{H_1}$ to $\overrightarrow{{\rm Pal}_7}$. Since $\overrightarrow{H}$ is a subgraph of $\overrightarrow{H_1}$, this leads to a contradiction. Therefore, $\overrightarrow{H}$ cannot  contain the configuration depicted in 
Figure~\ref{fig subcubic}(iii).

\medskip

\noindent \textit{Case 1:} Suppose that $v_3$ is a $-$-neighbor of $v_4$, that is, 
$\alpha=-$ and the arc $v_3v_4$ is present. In this case, $|S^-| < 5$ only if $f(u_4) = 1$ and $\beta_4 = -$, and we have $S^- = \{0,3,5,6\}$. The only way of having $N^{\beta_3}(f(u_3)) \cap S^- = \emptyset$. 
happens when $N^{\beta_3}(f(u_3)) = \{1,2,4\}$, which implies $f(u_3) = 0$ and $\beta_3 = +$. Since $f(u_3)=0$ and $f(u_4)=1$, $\overrightarrow{H_1}$ contains the arc $u_3u_4$. This means that the oriented cycle induced by $\{u_3, v_3, v_4, u_4\}$ has an odd number of backward and forward arcs. Since this property is invariant with respect to the push operation, this was also valid when $\overrightarrow{H_1}$ was constructed. Due to this construction, we obtain that $u_3$ and $u_4$ were adjacent in $\overrightarrow{H}$ itself. Moreover, by hypothesis, this implies that either $u_1$ and $u_4$ are adjacent, or $u_2$ and $u_3$ are adjacent. 

\begin{itemize}
	\item If $u_1$ and $u_4$ are adjacent, then we must have the arc $u_1u_4$, since $f(u_1)=0$ and $f(u_4)=1$. In this case, since $\overrightarrow{H}$ is cubic, the only neighbors of $u_4$ are $u_1,u_3$ and $v_4$, so we can modify $f$ by setting $f(u_4)= 2$, so that we now have $N^{\beta_3}(f(u_3)) \cap S^-\neq \emptyset$.

	\item If $u_2$ and $u_3$ are adjacent, then we must have the arc $u_3u_2$, since $f(u_3)=0$ and $f(u_2)=1$. Similarly, we modify $f$ by setting $f(u_3)= 4$, so that we now have $N^{\beta_3}(f(u_3)) \cap S^- \neq \emptyset$.
\end{itemize}

\medskip
\noindent \textit{Case 2:} Suppose that  $v_3$ is a $+$-neighbor of $v_4$, that is, $\alpha=+$ and the arc $v_4v_3$ is present. We have $|S^+| < 5$ only when $(f(u_4),\beta_4) = (1,+), (2,+), (0,-), (1, -)$ or $(6,-)$. In these cases, we respectively have $S^+ = \{0,3,5,6\}, \{0,1,5\}, \{0,1,5,6\}, \{0,1,2,4\}$ or $\{0,1,3,6\}$.

For $S^+ = \{0,1,5,6\}$, there exist no $f(u_3)$ and $\beta_3$ satisfying 
$N^{\beta_3}(f(u_3)) \cap S^+ = \emptyset$. Thus it is not possible to have $(f(u_4),\beta_4) = (0,-)$. 
Also note that in order to satisfy $N^{\beta_3}(f(u_3)) \cap S^+ = \emptyset$, for each value of $(f(u_4),\beta_4) = (1,+), (2,+), (1,-)$ or $(6,-)$, we must have $(f(u_3), \beta_3) = (0,+), (2,+), (0,-)$ or $(6,-)$. Since $u_3$ and $u_4$ are adjacent in $\overrightarrow{H_1}$, it is not possible to have $f(u_3)=f(u_4)$. Therefore, there exists $\gamma\in\{+,-\}$ such that $(f(u_4),\beta_4) = (1,\gamma)$, and thus $(f(u_3), \beta_3) = (0,\gamma)$. In particular, $\overrightarrow{H_1}$ contains the arc $u_3u_4$. This means that the oriented cycle induced by $\{u_3, v_3, v_4, u_4\}$ has an odd number of backward and forward arcs. Since this property is invariant with respect to the push operation, this was again also the case when $\overrightarrow{H_1}$ was constructed. This implies that $u_3$ and $u_4$ are adjacent in $\overrightarrow{H}$ itself. By hypothesis, either $u_1$ and $u_4$ are adjacent, or $u_2$ and $u_3$ are adjacent. 

\begin{itemize}
\item If $u_1$ and $u_4$ are adjacent, then we must have the arc $u_1u_4$, since $f(u_1)=0$ and $f(u_4)=1$. In this case, we modify $f$ by setting $f(u_4)= 2$, so that we now have $N^{\beta_4}(f(u_4)) \cap S^+ \neq \emptyset$.

	\item If $u_2$ and $u_3$ are adjacent, then we must have the arc $u_3u_2$, since $f(u_3)=0$ and $f(u_2)=1$. In this case, we modify $f$ by setting $f(u_3)= 4$, so that now we have $N^{\beta_3}(f(u_3)) \cap S^+ \neq \emptyset$.\qedhere
\end{itemize}
 \end{proof}

 \begin{lemma}\label{lem config iv}
 The configuration depicted in Figure~\ref{fig subcubic}(iv) cannot be contained in $\overrightarrow{H}$.
 \end{lemma}

 \begin{proof}
 Let $X$ be a graph on four vertices having a degree-$3$ vertex $x$ 
 with neighbors $a_1, a_2$ and  $b$. 
 Let $h: V(X) \rightarrow V(\overrightarrow{{\rm Pal}_7})$ be a function such that 
 $h(a_1) = 0$ and $h(a_2) = 1$. 
  Furthermore let $h^{*}: V(X) \rightarrow V(\overrightarrow{{\rm Pal}_7})$ be a function such that 
 $h^{*}(a_1) = h^{*}(a_2) = 0$. 
  Now let $\overrightarrow{X}$ be an orientation of $X$. Moreover, let $\overrightarrow{X'}$ be the orientation of $X$ obtained from $\overrightarrow{X}$ by pushing the vertex $x$. 
  Without loss of generality assume that the arc $a_1x$ is present in $\overrightarrow{X}$.
 Therefore $\overrightarrow{X}$ is of one of the following four types below. 
 For each type we list a number of observations following, notably, from the arc-transitivity of $\overrightarrow{{\rm Pal}_7}$.

 \medskip

\noindent \textit{Type 1:} $\overrightarrow{X}$ has the arcs $a_1x, a_2x$ and $bx$. 
Observe that for any $l \in V(\overrightarrow{{\rm Pal}_7}) \setminus \{2,4,6\}$  
it is possible to extend $h$ to  a homomorphism from $\overrightarrow{X}$ or $\overrightarrow{X'}$ 
to $\overrightarrow{{\rm Pal}_7}$ such that $h(b) = l$. 
Moreover, for any $l \in V(\overrightarrow{{\rm Pal}_7})$  
it is possible to extend $h^*$ to  a homomorphism from $\overrightarrow{X}$ or $\overrightarrow{X'}$ 
to $\overrightarrow{{\rm Pal}_7}$ such that $h^{*}(b) = l$. 

 \medskip

\noindent \textit{Type 2:} $\overrightarrow{X}$ has the arcs $a_1x, a_2x$ and $xb$. 
Observe that for any $l \in V(\overrightarrow{{\rm Pal}_7}) \setminus \{0,1\}$  
it is possible to extend $h$ to  a homomorphism from $\overrightarrow{X}$ or $\overrightarrow{X'}$ 
to $\overrightarrow{{\rm Pal}_7}$ such that $h(b) = l$. 
Moreover, for any $l \in V(\overrightarrow{{\rm Pal}_7}) \setminus \{0\}$  
it is possible to extend $h^*$ to  a homomorphism from $\overrightarrow{X}$ or $\overrightarrow{X'}$ 
to $\overrightarrow{{\rm Pal}_7}$ such that $h^{*}(b) = l$. 

 \medskip

\noindent \textit{Type 3:} $\overrightarrow{X}$ has the arcs $a_1x, xa_2$ and $bx$. 
Observe that for any $l \in V(\overrightarrow{{\rm Pal}_7}) \setminus \{1\}$  
it is possible to extend $h$ to a homomorphism from $\overrightarrow{X}$ or $\overrightarrow{X'}$ 
to $\overrightarrow{{\rm Pal}_7}$ such that $h(b) = l$.

 \medskip

\noindent \textit{Type 4:} $\overrightarrow{X}$ has the arcs $a_1x, xa_2$ and $xb$. 
Observe that for any $l \in V(\overrightarrow{{\rm Pal}_7}) \setminus \{0\}$  
it is possible to extend $h$ to  a homomorphism from $\overrightarrow{X}$ or $\overrightarrow{X'}$ 
to $\overrightarrow{{\rm Pal}_7}$ such that $h(b) = l$. 
 
 \medskip
 
  Assume now that $\overrightarrow{H}$ contains the configuration depicted in Figure~\ref{fig subcubic}(iv).
 Let $\overrightarrow{H_1}$ be the oriented graph obtained from 
 $\overrightarrow{H}$ by  deleting the vertex $b_0$,  
   $\overrightarrow{H_2}$ be the oriented graph obtained from $\overrightarrow{H_1}$ by 
    deleting the set of vertices $\{x_1, x_2, x_3\}$, and $\overrightarrow{H_3}$ be the oriented graph 
    obtained from 
 $\overrightarrow{H_2}$ by deleting the set of vertices $\{a_{11}, a_{12}\}$.
 By minimality, $\overrightarrow{H_1}$ admits a pushable homomorphism $f_1$ to 
$\overrightarrow{{\rm Pal}_7}$. Up to replacing $\overrightarrow{H}$ by a push equivalent oriented graph, we may assume that $f_1$ is an oriented homomorphism. Let $f_2$ and $f_3$ be the restriction of $f_1$ to $\overrightarrow{H_2}$ and $\overrightarrow{H_3}$, respectively. 

In what follows, we say that the vertex $x_i$  \textit{is Type-$j$},
for some $j\in\{1,2,3,4\}$, if the oriented graph induced by $\{a_{i1}, a_{i2}, x_i, b_0\}$ is the same as  the Type-$j$ orientation of $\overrightarrow{X}$ or $\overrightarrow{X'}$ where $x_i$ plays the role of $x$ and $b_0$ plays the role of $b$   for $i \in \{1,2,3\}$ and $j \in \{1,2,3,4\}$. 
Note that if $x_i$ is Type-1 (Type-3, respectively), then pushing the vertex $b_0$ will turn it to 
Type-2 (Type-4, respectively), and vice versa. Now we want to extend $f_2$ or $f_3$ to 
a pushable homomorphism from $\overrightarrow{H}$ to $\overrightarrow{{\rm Pal}_7}$.
Note that from the above we know that 
if $x_i$ is Type-1 then it may forbid at most three values for $f(b_0)$,
if $x_i$ is Type-2 then it may forbid at most two values for $f(b_0)$, and
if $x_i$ is Type-3 or Type-4 then it may forbid at most one value for $f(b_0)$.
Moreover, we may push $b_0$ to ensure that at most one vertex among $\{x_1, x_2, x_3\}$ is Type-1. 
Therefore, if any of them is Type-3 or Type-4, then we will be able to extend 
$f_2$  to 
a pushable homomorphism from $\overrightarrow{H}$ to $\overrightarrow{{\rm Pal}_7}$.

The only bad situation is to have one Type-1 vertex and two Type-2 vertices among $\{x_1, x_2, x_3\}$. 
We here consider two cases:

\medskip

\noindent \textit{Case 1:}  If $x_1$ is Type-1, then we simply push $a_{11}$ and $a_{12}$ to make it 
Type-2. Then we are able to extend 
$f_3$  to 
a pushable homomorphism from $\overrightarrow{H}$ to $\overrightarrow{{\rm Pal}_7}$.

\medskip

\noindent \textit{Case 2:}  Assume that $x_2$ is Type-1 and $x_1,x_3$ are Type-2. Without loss of generality assume that $f_1(a_{21}) = 0$ and $f_1(a_{22}) = 1$ (since if $f_1(a_{21})=f_1(a_{22})$, then $x_2$ forbids only two values for $f(b_0)$). Therefore, $x_2$ forbids the values $2, 4, 6$ for $f(b_0)$.

By definition of type 2, $x_1$ (resp. $x_3$) forbids the values $f_1(a_{11}),f_1(a_{12})$ (resp. $f_1(a_{31}),f_1(a_{32})$). We can then extend $f_2$ to a pushable homomorphism from $\overrightarrow{H}$ to $\overrightarrow{{\rm Pal}_7}$ unless $\{f_1(a_{11}), f_1(a_{12}), f_1(a_{31}), f_1(a_{32})\} = \{0,1,3,5\}$. 

For $S \subseteq \{a_{11},a_{12}\}$, we denote by $\overrightarrow{H_2}(S)$ the graph obtained from $\overrightarrow{H_2}$ by pushing the vertices of $S$. Note that for each choice of $S$, we can extend $f_3$ to an oriented homomorphism $f_S$ from $\overrightarrow{H_2}(S)$ to $\overrightarrow{{\rm Pal}_7}$. We may even ensure that for every vertex $v\notin S$, we have $f_S(v)=f_2(v)$ and, for $v\in\{a_{11},a_{12}\}$,  $f_{\{a_{11},a_{12}\}}(v)=f_{\{v\}}(v)$. We show that there exists a choice of $S$ such that $f_S$ can be extended to $\overrightarrow{H}(S)$. 

Let $S=\{a_{11},a_{12}\}$, and observe that $x_1$ has Type-1 in $\overrightarrow{H}(S)$.
\begin{itemize}
    \item If $(f_1(a_{11}), f_1(a_{12})) = (f_S(a_{12}), f_S(a_{11}))$, then due to the property of Type-1 vertices, $x_1$ does not forbid the values $f_S(a_{11})=f_1(a_{12})$ and $f_S(a_{12})=f_1(a_{11})$ for $b_0$ anymore. In particular, we can extend $f_S$ to a pushable homomorphism from $\overrightarrow{H}$ to $\overrightarrow{{\rm Pal}_7}$.
    \item If $f_1(a_{11})\neq f_S(a_{12})$, then we claim that $f_{\{a_{12}\}}$ can be extended to $\overrightarrow{H}(\{a_{12}\})$. Indeed, in $\overrightarrow{H}(\{a_{12}\})$, $x_1$ is Type-3 or Type-4. In particular, we can extend $f_{\{a_{12}\}}$ unless $f_{\{a_{12}\}}(a_{12})=f_{\{a_{12}\}}(a_{11})$. However, note that we have $f_S(a_{12})=f_{\{a_{12}\}}(a_{12})$ and $f_{\{a_{12}\}}(a_{11})=f_2(a_{11})=f_1(a_{11})$ by definition of $f_S$. This is impossible by hypothesis.
    \item Otherwise, we have $f_1(a_{12})\neq f_S(a_{11})$. Similarly to the previous item, we claim that $f_{\{a_{11}\}}$ can be extended to $\overrightarrow{H}(\{a_{11}\})$. Indeed, $x_1$ is again Type-3 or Type-4, so $f_{\{a_{11}\}}$ can be extended unless $f_{\{a_{11}\}}(a_{11})=f_{\{a_{11}\}}(a_{12})$, which is again impossible by hypothesis.
\end{itemize}
In each case, we can thus extend $f_3$ to a pushable homomorphism from $\overrightarrow{H}$ to $\overrightarrow{{\rm Pal}_7}$, which concludes the proof.
 \end{proof}

\section{Proof of Theorem~\ref{th mad3}}\label{sec 3}
 
The lower bound follows from the existence of oriented graphs with maximum average degree less than $3$ and pushable chromatic number~$5$,
such as the one depicted in Figure~\ref{orientable coloring girth 4}(ii).
 To prove the upper bound of  Theorem~\ref{th mad3}, we show below that every oriented graph with maximum average degree less than $3$ admits a pushable homomorphism to the Paley tournament 
 $\overrightarrow{{\rm Pal}_7}$ on seven vertices.
Towards a contradiction, let us assume that there exists $\overrightarrow{H}$, a minimum (with respect to number of vertices) oriented graph having ${\rm mad}(\overrightarrow{H}) < 3$ that does not admit a pushable homomorphism to 
$\overrightarrow{{\rm Pal}_7}$.   Note that $\overrightarrow{H}$ must be connected due to the minimality condition. 
We prove below that, by minimality, none of the configurations depicted in Figure~\ref{fig mad3} can appear in $\overrightarrow{H}$.

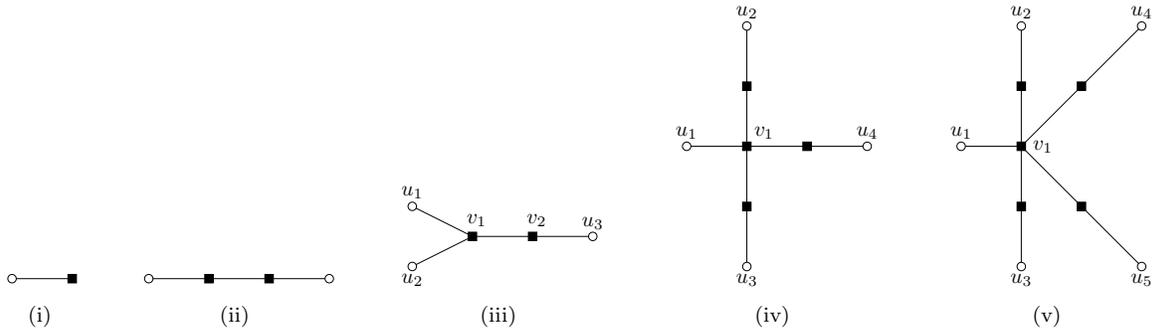
\begin{figure} 
\centering
\subfloat[]{
\scalebox{0.8}{
\begin{tikzpicture}
\path (0,6) edge (1,6);
\filldraw [black,draw,fill=white] (0,6) circle (2pt) {node[below]{}};
\filldraw [black] ([xshift=-2pt,yshift=-2pt]1,6) rectangle ++(4pt,4pt) {node[above]{}};
\end{tikzpicture}
}}
\hspace{10pt}
\subfloat[]{
\scalebox{0.8}{
\begin{tikzpicture}
\draw[-] (2,6) -- (5,6);
\filldraw [black,draw,fill=white] (2,6) circle (2pt) {node[below]{}};
\filldraw [black] ([xshift=-2pt,yshift=-2pt]3,6) rectangle ++(4pt,4pt) {node[above]{}};
\filldraw [black,draw,fill=white] (5,6) circle (2pt) {node[below]{}};
\filldraw [black] ([xshift=-2pt,yshift=-2pt]4,6) rectangle ++(4pt,4pt) {node[above]{}};
\end{tikzpicture}
}}
\hspace{10pt}
\subfloat[]{
\scalebox{0.8}{
\begin{tikzpicture}
\draw[-] (7,6) -- (9,6);
\draw[-] (7,6) -- (6,6.5);
\draw[-] (7,6) -- (6,5.5);
\filldraw [black,draw,fill=white] (6,6.5) circle (2pt) {node[above]{$u_1$}};
\filldraw [black,draw,fill=white] (6,5.5) circle (2pt) {node[below]{$u_2$}};
\filldraw [black,draw,fill=white] (9,6) circle (2pt) {node[above]{$u_3$}};
\filldraw [black] ([xshift=-2pt,yshift=-2pt]7,6) rectangle ++(4pt,4pt) {node[above,yshift=-1pt]{$v_1$}};
\filldraw [black] ([xshift=-2pt,yshift=-2pt]8,6) rectangle ++(4pt,4pt) {node[above,yshift=-1pt]{$v_2$}};
\end{tikzpicture}
}}
\hspace{10pt}
\subfloat[]{
\scalebox{0.8}{
\begin{tikzpicture}
\draw (1,2) -- (4,2);
\draw (2,0) -- (2,4);
\filldraw [black,draw,fill=white] (1,2) circle (2pt) {node[above]{$u_1$}};
\filldraw [black,draw,fill=white] (2,4) circle (2pt) {node[above]{$u_2$}};
\filldraw [black,draw,fill=white] (2,0) circle (2pt) {node[below]{$u_3$}};
\filldraw [black,draw,fill=white] (4,2) circle (2pt) {node[above]{$u_4$}};
\filldraw [black] ([xshift=-2pt,yshift=-2pt]2,1) rectangle ++(4pt,4pt) {node[above]{}};
\filldraw [black] ([xshift=-2pt,yshift=-2pt]2,2) rectangle ++(4pt,4pt) {node[above right,yshift=-2pt,xshift=-2pt]{$v_1$}};
\filldraw [black] ([xshift=-2pt,yshift=-2pt]2,3) rectangle ++(4pt,4pt) {node[above]{}};
\filldraw [black] ([xshift=-2pt,yshift=-2pt]3,2) rectangle ++(4pt,4pt) {node[above]{}};
\end{tikzpicture}
}}
\hspace{10pt}
\subfloat[]{
\scalebox{0.8}{
\begin{tikzpicture}
\draw[-] (5,2) -- (6,2);
\draw[-] (6,0) -- (6,4);
\draw[-] (6,2) -- (8,4);
\draw[-] (6,2) -- (8,0);
\filldraw [black,draw,fill=white] (5,2) circle (2pt) {node[above]{$u_1$}};
\filldraw [black,draw,fill=white] (6,4) circle (2pt) {node[above]{$u_2$}};
\filldraw [black,draw,fill=white] (6,0) circle (2pt) {node[below]{$u_3$}};
\filldraw [black,draw,fill=white] (8,4) circle (2pt) {node[above]{$u_4$}};
\filldraw [black,draw,fill=white] (8,0) circle (2pt) {node[below]{$u_5$}};
\filldraw [black] ([xshift=-2pt,yshift=-2pt]6,1) rectangle ++(4pt,4pt) {node[above]{}};
\filldraw [black] ([xshift=-2pt,yshift=-2pt]6,2) rectangle ++(4pt,4pt) {node[right,yshift=-2pt]{$v_1$}};
\filldraw [black] ([xshift=-2pt,yshift=-2pt]6,3) rectangle ++(4pt,4pt) {node[above]{}};
\filldraw [black] ([xshift=-2pt,yshift=-2pt]7,1) rectangle ++(4pt,4pt) {node[above]{}};
\filldraw [black] ([xshift=-2pt,yshift=-2pt]7,3) rectangle ++(4pt,4pt) {node[above]{}};
\end{tikzpicture}
}}

\caption{Configurations needed for proving Theorem~\ref{th mad3}. Black square vertices are vertices whose full neighborhood is part of the configuration. 
White circle vertices are vertices that might have other neighbors outside the configuration, or be equal.}\label{fig mad3}

\end{figure}	

\begin{lemma}\label{lemma:configs-mad}
None of the configurations (i)-(v) depicted in Figure~\ref{fig mad3} can be 
contained in $\overrightarrow{H}$. 
\end{lemma}

\begin{proof}
\noindent $(i)$ Suppose $\overrightarrow{H}$ has a degree-$1$ vertex $u$ with neighbor $v$, where $u$ is an $\alpha$-neighbor of $v$ for $\alpha\in\{+,-\}$.
By minimality, $\overrightarrow{H} - \{u\}$ admits a pushable homomorphism $f$ to 
$\overrightarrow{{\rm Pal}_7}$. 
We can extend $f$ to $\overrightarrow{H}$ by setting 
\begin{equation*}
f(u) =
\begin{cases}
f(v)+1 & \text{ if } \alpha = + \\
f(v)+3 & \text{ if } \alpha = -,
\end{cases}
\end{equation*}
a contradiction.
Therefore, $\overrightarrow{H}$ cannot  contain the configuration depicted in Figure~\ref{fig mad3}(i).

\medskip

\noindent $(ii)$  Assume $\overrightarrow{H}$ has two adjacent degree-$2$ vertices $u$ and $v$.
By minimality, $\overrightarrow{H} - \{u,v\}$ admits a pushable homomorphism $f$ to 
$\overrightarrow{{\rm Pal}_7}$. 
According to Equation (\ref{eq second neighborhood}), $f$ can be extended to a pushable homomorphism from $\overrightarrow{H}$
to $\overrightarrow{{\rm Pal}_7}$, a contradiction.
Therefore, $\overrightarrow{H}$ cannot  contain the configuration depicted in Figure~\ref{fig mad3}(ii).

\medskip

\noindent $(iii)$ Assume $\overrightarrow{H}$ contains the configuration depicted in Figure~\ref{fig mad3}(iii) (where, here and further, we deal with the vertices of a configuration through the notation introduced in the corresponding figure). 
By minimality, $\overrightarrow{H} - \{v_2\}$ admits a pushable homomorphism $f$ to 
$\overrightarrow{{\rm Pal}_7}$. 
Suppose that 
$v_1 \in N^{(\alpha_1, \alpha_2)}(u_1,u_2)$. 
 If $f(v_1) = f(u_3)$, then push $v_1$ and update $f(v_1)$ to some value 
 in $N^{(\bar{\alpha}_1,\bar{\alpha}_2)}(f(u_1),f(u_2))$ (this is possible since $\overrightarrow{{\rm Pal}_7}$ has property $P(2,2)$). Note that the resulting updated $f$ remains a
 pushable homomorphism. Additionally, now we have  $f(v_1) \neq f(u_3)$. It is now possible to extend $f$ to a pushable homomorphism from $\overrightarrow{H}$
to $\overrightarrow{{\rm Pal}_7}$ due to Equation (\ref{eq first neighborhood}) on page~\pageref{eq first neighborhood}, a contradiction.
Therefore, $\overrightarrow{H}$ cannot  contain the configuration depicted in Figure~\ref{fig mad3}(iii).

\medskip

\noindent $(iv)-(v)$  Assume $\overrightarrow{H}$ contains one of the configurations depicted in Figure~\ref{fig mad3}(iv) and (v). Let $S$ be the set of black vertices of the configuration. By minimality, $\overrightarrow{H} - S$ admits a pushable homomorphism 
$f$ to 
$\overrightarrow{{\rm Pal}_7}$. Now, push $v_1$, if needed, so that at most two of 
the $2$-paths connecting $v_1$ to the $u_i$'s are directed $2$-paths. 
Without loss of generality, assume that, in the worst-case scenario, $u_2$ and $u_3$ are the vertices that are connected by a directed $2$-path to $v_1$. 
Assume $v_1$ is an $\alpha$-neighbor of $u_1$ for some $\alpha \in \{+,-\}$.
Choose a vertex $i \in N^{\alpha}(f(u_1)) \setminus \{f(u_2), f(u_3)\}$ and assign $f(v_1) = i$. 
Now we are able to  extend $f$ to a pushable homomorphism from $\overrightarrow{H}$
to $\overrightarrow{{\rm Pal}_7}$ due to Equation~(\ref{eq first neighborhood}) on page~\pageref{eq first neighborhood}, a contradiction.
Therefore, $\overrightarrow{H}$ cannot contain the configurations depicted in Figure~\ref{fig mad3}(iv) and (v).
\end{proof}
 
 We are now ready to prove Theorem~\ref{th mad3}, which we do using the so-called \textit{discharging method}.

\begin{proof}[Proof of Theorem~\ref{th mad3}]
Let us assign the charge ${\rm ch}(v) = d(v)$  to each vertex $v$ of $\overrightarrow{H}$. 
 Since ${\rm mad}(\overrightarrow{H}) < 3$, the total sum of the charges is strictly less than $3|V(\overrightarrow{H})|$, that is, 
 $$\sum_{v \in V(\overrightarrow{H})} {\rm ch}(v) < 3|V(\overrightarrow{H})|.$$ 
 
 Now apply the following discharging procedure: 
Every vertex $v$ of $\overrightarrow{H}$
with degree at least
$4$ sends $1/2$ to each of its neighbors with degree~$2$. 
We show in the following that for every vertex $v$ the resulting charge ${\rm ch}^*(v)$ is at least $3$, which
contradicts the assumption ${\rm mad}(H)< 3$. 
We consider the vertices $v$ accordingly to their degree, which satisfies $d(v)>1$ by Lemma~\ref{lemma:configs-mad}(i).

\begin{itemize}
\item $d(v) = 2$. Since $\overrightarrow{H}$ does not contain the configurations depicted in Figures~\ref{fig mad3}(ii) and~\ref{fig mad3}(iii), the neighbors of $v$ have degree at least 4 and thus $v$ does not send any charge. Furthermore, $v$ receives exactly $2 \times 1/2 = 1$. Thus, ${\rm ch}^*(v) = 2+1=3$.

\item $d(v)=3$. Since $\overrightarrow{H}$ does not contain the configuration depicted in Figure~\ref{fig mad3}(iii), $v$ does not send any charge. Furthermore, $v$ does not receive any charge. 
Therefore, we have 
 ${\rm ch}^*(v) =  3$.

\item $d(v) = 4$. Since $\overrightarrow{H}$ does not contain the configuration depicted in Figure~\ref{fig mad3}(iv), $v$ sends at most $2 \times 1/2 = 1$. Therefore, we have 
 ${\rm ch}^*(v) \geq  4-1 = 3$.

 \item $d(v) = 5$. Since $\overrightarrow{H}$ does not contain the configuration depicted in Figure~\ref{fig mad3}(v), $v$ sends at most $3 \times 1/2 = 3/2$. Therefore, we have 
 ${\rm ch}^*(v) \geq  5-3/2 = 7/2 > 3$.
 
  \item $d(v) = k \geq 6$.  $v$ sends at most $k \times 1/2 = k/2$ charges. Therefore, we have 
 ${\rm ch}^*(v) \geq  k-k/2 = k/2 \geq 6/2 = 3$.
\end{itemize}

 Therefore, every vertex $v$ of $\overrightarrow{H}$ gets final charge ${\rm ch}^*(v)$ at least 3. Hence
  $$3|V(\overrightarrow{H})| > \sum_{v \in V(\overrightarrow{H})} {\rm ch}(v)=\sum_{v \in V(\overrightarrow{H})} {\rm ch}^*(v) \geq 3|V(\overrightarrow{H})|,$$
  since no charge was created after assigning the initial charges, which is a contradiction.
  Thus every oriented graph with maximum average degree less than $3$ admits a pushable 
  homomorphism to $\overrightarrow{{\rm Pal}_7}$.  
\end{proof}

 \section{Conclusions and perspectives} \label{section:conclusion}
 
In this work, we have studied the pushable chromatic number of several classes of graphs with degree constraints.
We have provided bounds for graphs with large maximum degree $\Delta$ (Theorem~\ref{Push_chrom_th main}), graphs with maximum degree $\Delta \leq 3$ (Theorem~\ref{th subcubic}),
and graphs with maximum average degree less than~$3$ (Theorem~\ref{th mad3}). None of our results is tight however, and a natural direction for further work could thus be to tighten our bounds.
In particular, we wonder whether there exist subcubic graphs or graphs with maximum average degree less than~$3$ with pushable chromatic reaching the upper bounds we have established.
Let us mention that we first checked the proof of Theorem~\ref{th subcubic} through computer programs (before coming up with the presented matrices), 
and that we did not find any tournament on six vertices for which all configurations in Figure~\ref{fig subcubic} are reducible. 
Also, although we managed to generate many graphs with maximum average degree less than~$3$ (planar graphs with girth at least~$6$, respectively) and check their pushable chromatic number via computer programs,
we were not able to spot one with pushable chromatic~$6$ ($5$, respectively). These two facts might be good hints regarding the maximum value of the pushable chromatic number of these families of graphs.
 
Another interesting direction for further research on the topic could be to generalize our results to graphs with given maximum degree $\Delta$ more than~$3$, 
graphs with given maximum average degree, and planar graphs with given girth. In other words, we wonder how these graph parameters influence the pushable chromatic number.
We would be quite interested, for instance, in having bounds for graphs with maximum degree $\Delta$ at most~$4$.
 
Finally, several recent works have established that, when it comes to coloring, pushable graphs and signed graphs sometimes have very comparable behaviors.
Let us recall that a \textit{signed graph} is a graph in which each edge is either positive or negative, and that comes with a vertex-resigning operation which consists in switching the sign of all edges incident with a vertex.
It would be interesting to know if, in general, graphs with degree constraints have their pushable chromatic number and signed chromatic number behaving the same.
We will propose a study of this very question, inspired from our results in the current work, in a forthcoming paper.

\bibliographystyle{abbrv}
\bibliography{NSS14}

\end{document}